\documentclass[11pt,a4paper]{article}

\usepackage[margin=1in]{geometry}
\usepackage[T1]{fontenc}
\usepackage{amsmath,amssymb,amsthm,mathtools}
\usepackage{microtype}
\usepackage{enumitem}
\usepackage{xcolor}
\usepackage{aliascnt}
\usepackage{hyperref}
\usepackage[capitalise,noabbrev]{cleveref}

\hypersetup{
  colorlinks=true,
  citecolor=blue!60!black,
  linkcolor=blue!60!black,
  urlcolor=blue!60!black
}
\allowdisplaybreaks

\newtheorem{theorem}{Theorem}[section]
\newaliascnt{lemma}{theorem}
\newtheorem{lemma}[lemma]{Lemma}
\aliascntresetthe{lemma}
\newaliascnt{proposition}{theorem}
\newtheorem{proposition}[proposition]{Proposition}
\aliascntresetthe{proposition}
\newaliascnt{corollary}{theorem}
\newtheorem{corollary}[corollary]{Corollary}
\aliascntresetthe{corollary}
\theoremstyle{definition}
\newaliascnt{problem}{theorem}
\newtheorem{problem}[problem]{Problem}
\aliascntresetthe{problem}
\newaliascnt{definition}{theorem}

\aliascntresetthe{definition}
\newaliascnt{remark}{theorem}

\aliascntresetthe{remark}

\newcommand{\cI}{\mathcal I}
\newcommand{\R}{\mathbb R}
\newcommand{\C}{\mathbb C}
\newcommand{\E}{\mathbb E}
\renewcommand{\Pr}{\mathbb P}
\newcommand{\one}{\mathbf 1}
\newcommand{\iu}{\mathrm i}
\newcommand{\e}{\mathrm e}
\newcommand{\TV}{d_{\mathrm{TV}}}
\newcommand{\abs}[1]{\left\lvert #1\right\rvert}
\newcommand{\norm}[1]{\left\lVert #1\right\rVert}
\newcommand{\set}[1]{\left\{#1\right\}}

\DeclareMathOperator{\poly}{poly}

\title{Deterministic Approximation of the Total Variation Distance Between Spin Systems}
\author{Zelin Li \footnote{School of Computer Science and Technology, University of Science and Technology of China, email: \texttt{meguru@mail.ustc.edu.cn}.} \and Minji Yang \footnote{School of Computing and Data Science, The University of Hong Kong, email: \texttt{ymjessen02@connect.hku.hk}.}}
\date{}

\begin{document}

\maketitle

\begin{abstract}
We study deterministic relative approximation of the total variation distance
between two Gibbs distributions induced by spin systems on the same
bounded-degree graph.  For the hard-core model, we give a deterministic
polynomial-time $\varepsilon$-relative-error approximation when both external-field vectors lie in
$[b,(1-\eta)\lambda_{\mathrm c}(\Delta)]^V$, where $b>0$ and
$\eta\in(0,1)$ are fixed and $\lambda_{\mathrm c}(\Delta)$ is the hard-core
uniqueness threshold.  For the Ising model, we
obtain deterministic polynomial-time algorithms in two settings: the
ferromagnetic Lee--Yang regime
and the antiferromagnetic correlation-decay regime.  We
develop a new deterministic framework that reduces this task to estimating
suitably chosen partition functions.
\end{abstract}
\section{Introduction}
\label{sec:introduction}

The total variation distance is a fundamental measure of the difference
between two probability distributions.  For distributions $\mu$ and $\nu$ on
a finite set $\Omega$, it is defined by
\[
  \TV(\mu,\nu)
  =
  \frac{1}{2}\sum_{\sigma\in\Omega}
  \abs{\mu(\sigma)-\nu(\sigma)}.
\]

It can be evaluated in $O(\abs{\Omega})$ time when both probability tables are
given explicitly.  The problem becomes interesting when both high-dimensional
distributions admit succinct descriptions.  This setting, however, already
gives rise to strong complexity barriers.  Sahai and Vadhan~\cite{sahaiVadhan2003complete}
proved that the statistical-difference problem for circuit-generated
distributions is complete for statistical zero knowledge.  Bhattacharyya et al.~\cite{bhattacharyyaEtAl2023approximating}
showed that exact TV-distance computation is $\#\mathrm{P}$-complete even for
Boolean product distributions, and that relative approximation is
$\mathsf{NP}$-hard for Bayesian networks.  More recently, Bhattacharyya et al.~\cite{bhattacharyyaEtAl2025computational}
proved that, unless $\mathsf{NP}\subseteq\mathsf{RP}$, no FPRAS exists for the
TV-distance problem over arbitrary Ising-model inputs.

These barriers make structured distribution families the natural setting in
which to design approximation algorithms.  Product distributions provide the
first example: Bhattacharyya et al.~\cite{bhattacharyyaEtAl2023approximating} obtained an FPTAS for a
restricted Boolean case.  This restriction was later removed, first through
an FPRAS for general product distributions~\cite{fengGuoJerrumWang2023product}
and then through an FPTAS~\cite{fengLiuLiu2024deterministically}.  For two
Bayesian networks defined over the same DAG, Bhattacharyya et
al.~\cite{bhattacharyyaEtAl2024inference} reduced relative TV-distance
approximation to probabilistic inference, yielding an FPRAS whenever exact
inference is efficient.  For mixtures of product distributions,
Bhattacharyya et al.~\cite{bhattacharyyaEtAl2025computational} gave a
polynomial-time algorithm to decide whether the TV distance is zero, and
Feng et al.~\cite{feng2026computing} further gave an FPRAS for relative
TV-distance approximation for a constant number of mixture components.

For spin systems, Feng et al.~\cite{pmlr-v291-feng25a} reduced relative
TV-distance approximation to sampling and approximate counting.  The resulting
FPRAS applies to the hard-core model throughout the constant-gap uniqueness
regime and, under a constant marginal lower bound, to Ising models in the
spectral, ferromagnetic consistent-field, and antiferromagnetic uniqueness
regimes.  Feng and Fu~\cite{feng_et_al:LIPIcs.ITCS.2026.59} gave
randomized polynomial-time algorithms for relative approximation
of several $f$-divergences between Ising models
in suitable parameter regimes.

In this paper, we study deterministic polynomial-time
$\varepsilon$-relative-error approximation of the TV distance between two
bounded-degree spin systems, including the hard-core model and the Ising
model.  Our framework reduces the problem to deterministically approximating
suitable partition functions at real and complex activities.

\subsection{Main results}
\label{sec:problem-results}

\subsubsection{Models and problems}
\label{sec:models-problems}

Let $G=(V,E)$ be a finite simple graph, write $n=\abs V$, and set
$\Sigma_V=\{-1,+1\}^V$.  A binary spin system $\mathsf S$ on $G$ assigns a
nonnegative weight $w_{\mathsf S}(\sigma)$ to each configuration
$\sigma\in\Sigma_V$.  If at least one configuration has positive weight, its
partition function and Gibbs distribution are defined by
\[
  Z_{\mathsf S}
  =\sum_{\tau\in\Sigma_V}w_{\mathsf S}(\tau),
  \qquad
  \mu_{\mathsf S}(\sigma)
  =\frac{w_{\mathsf S}(\sigma)}{Z_{\mathsf S}},
  \quad \sigma\in\Sigma_V.
\]
We consider the following two spin systems.

\paragraph{Hard-core model.}
Let $\lambda=(\lambda_v)_{v\in V}\in\R_{\ge0}^V$ be a vector of vertex
activities.  For $\sigma\in\Sigma_V$, define
\[
  I(\sigma)=\set{v\in V:\sigma_v=+1}.
\]
The configuration $\sigma$ is feasible if $I(\sigma)$ is an independent set
of $G$.  The hard-core model $(G,\lambda)$ has weight function
\[
  w_\lambda(\sigma)
  =
  \begin{cases}
    \displaystyle\prod_{v\in I(\sigma)}\lambda_v,
      & \text{if $I(\sigma)$ is an independent set of $G$},\\[4pt]
    0, & \text{otherwise}.
  \end{cases}
\]

\paragraph{Ising model.}
Let $\beta>0$ be an edge activity and let
$\lambda=(\lambda_v)_{v\in V}\in\R_{>0}^V$ be a vector of vertex activities.
For $\sigma\in\Sigma_V$, let
\[
  D_G(\sigma)=\set{\{u,v\}\in E:\sigma_u\ne\sigma_v}.
\]
The Ising model $(G,\beta,\lambda)$ has weight function
\[
  w_{\beta,\lambda}(\sigma)
  =\beta^{\abs{D_G(\sigma)}}
    \prod_{v:\,\sigma_v=+1}\lambda_v,
  \qquad \sigma\in\Sigma_V.
\]
For $z\in\C^V$, the notation $Z_G(\beta,z)$ denotes the same partition
function polynomial with the vertex activities replaced by $z$.  Under this
parametrization, the model is ferromagnetic for $0<\beta<1$, has independent
spins for $\beta=1$, and is antiferromagnetic for $\beta>1$.

We compare two systems of the same type on the same graph and consider the
following two problems.

\begin{samepage}
\begin{problem}[Hard-core TV-distance approximation]
\label{prob:hardcore-tv}
\leavevmode
\begin{itemize}[leftmargin=2.4em]
  \item \emph{Input:} A finite simple graph $G=(V,E)$, two activity vectors
  $\lambda^\mu,\lambda^\nu\in\R_{>0}^V$, and an error bound
  $\varepsilon\in(0,1)$.  These specify two hard-core Gibbs distributions
  $\mu=\mu_{G,\lambda^\mu}$ and $\nu=\mu_{G,\lambda^\nu}$.

  \item \emph{Output:} A nonnegative real number $\widehat d$ such that
  \[
    (1-\varepsilon)\TV(\mu,\nu)
    \le \widehat d
    \le (1+\varepsilon)\TV(\mu,\nu).
  \]
\end{itemize}
\end{problem}
\end{samepage}

\begin{samepage}
\begin{problem}[Ising TV-distance approximation]
\label{prob:ising-tv}
\leavevmode
\begin{itemize}[leftmargin=2.4em]
  \item \emph{Input:} A finite simple graph $G=(V,E)$, an edge activity
  $\beta>0$, two vertex-activity vectors
  $\lambda^\mu,\lambda^\nu\in\R_{>0}^V$, and an error bound
  $\varepsilon\in(0,1)$.  These specify two Ising Gibbs distributions
  $\mu=\mu_{G,\beta,\lambda^\mu}$ and
  $\nu=\mu_{G,\beta,\lambda^\nu}$.

  \item \emph{Output:} A nonnegative real number $\widehat d$ such that
  \[
    (1-\varepsilon)\TV(\mu,\nu)
    \le \widehat d
    \le (1+\varepsilon)\TV(\mu,\nu).
  \]
\end{itemize}
\end{problem}
\end{samepage}

\subsubsection{Our results}

\paragraph{Hard-core model.}

For $\Delta\ge 3$, write
\begin{equation*}
  \lambda_{\mathrm c}(\Delta)
  =\frac{(\Delta-1)^{\Delta-1}}{(\Delta-2)^\Delta}
\end{equation*}
for the hard-core uniqueness threshold on the infinite $\Delta$-regular tree.

\begin{theorem}[Main theorem]
\label{thm:uniqueness-main}
Fix $\Delta\ge3$, $\eta\in(0,1)$, and $b>0$, and let
$r=(1-\eta)\lambda_{\mathrm c}(\Delta)$.  If $b<r$, then there is a
deterministic algorithm for
\cref{prob:hardcore-tv} on instances in which $G$ has maximum degree at most
$\Delta$ and
\[
  \lambda^\mu,\lambda^\nu\in[b,r]^V.
\]
Its running time is
\[
  \poly_{\Delta,\eta,b}\left(n,\frac1\varepsilon\right),
  \qquad n=\abs V.
\]
An explicit running-time bound is given in
\eqref{eq:hardcore-total-time}.
The output is $0$ if and only if $\lambda^\mu=\lambda^\nu$.
\end{theorem}

\paragraph{Ising model.}

The first result for the Ising model uses the
multivariate Lee--Yang zero-free region
\cite{leeYang1952statistical,liuSinclairSrivastava2017ising}.

\begin{theorem}[Lee--Yang zero-free regime]
\label{thm:ising-lee-yang}
Fix an integer $\Delta\ge1$ and constants $0<\beta\le1$ and
$0<b<r<1$.  There is a deterministic algorithm for
\cref{prob:ising-tv} on instances in which $G$ has maximum degree at most
$\Delta$ and either
\[
  \lambda^\mu,\lambda^\nu\in[b,r]^V
\]
or
\[
  \lambda^\mu,\lambda^\nu\in[r^{-1},b^{-1}]^V.
\]
Its running time is
\[
  \poly_{\Delta,\beta,b,r}\left(n,\frac1\varepsilon\right),
  \qquad n=\abs V.
\]
An explicit running-time bound is given in
\eqref{eq:lee-yang-total-time}.
The output is $0$ if and only if $\lambda^\mu=\lambda^\nu$.
\end{theorem}

We next describe the antiferromagnetic parameter region.  Fix
$\Delta\ge3$ and $\beta>1$ with
$\beta\ne\Delta/(\Delta-2)$.  Following Sinclair et
al.~\cite{sinclairSrivastavaThurley2012approximation} and Liu et
al.~\cite{liuSinclairSrivastava2019correlation}, let
\begin{equation*}
  \mathcal U_{\beta,\Delta}
  =
  \begin{cases}
    (0,\infty),
      & 1<\beta<\dfrac{\Delta}{\Delta-2},\\[6pt]
    (0,\lambda_{\mathrm c}(\beta,\Delta)^{-1})
      \mathbin{\cup}
      (\lambda_{\mathrm c}(\beta,\Delta),\infty),
      & \beta>\dfrac{\Delta}{\Delta-2},
  \end{cases}
\end{equation*}
be the strict correlation-decay region on the infinite
$\Delta$-regular tree, where $\lambda_{\mathrm c}(\beta,\Delta)>1$ is the
critical external field in the second case. For the critical-field
characterization, see \cite[pp.~254--255]{georgii2011gibbs}
and its complete definition in
\cite{sinclairSrivastavaThurley2012approximation}.

\begin{theorem}[Antiferromagnetic correlation-decay regime]
\label{thm:ising-correlation-decay}
Fix an integer $\Delta\ge3$ and a constant $\beta>1$ with
$\beta\ne\Delta/(\Delta-2)$.  Let $[\lambda_-,\lambda_+]$ be a fixed
nonempty closed interval contained in one connected component of
$\mathcal U_{\beta,\Delta}$.  There is a deterministic algorithm for
\cref{prob:ising-tv} on instances in which $G$ has maximum degree at most
$\Delta$ and
\[
  \lambda^\mu,\lambda^\nu\in[\lambda_-,\lambda_+]^V.
\]
Its running time is
\[
  \poly_{\Delta,\beta,\lambda_-,\lambda_+}
  \left(n,\frac1\varepsilon\right),
  \qquad n=\abs V.
\]
An explicit running-time bound is given in
\eqref{eq:ising-cd-total-time}.
The output is $0$ if and only if $\lambda^\mu=\lambda^\nu$.
\end{theorem}

\subsection{Related work}

\paragraph{TV-distance approximation.}
Earlier work~\cite{bhattacharyyaEtAl2020learning,chenKiefer2014total,
kiefer2018computing,canonneRubinfeld2014testing} studied algorithms and
hardness for additive-error approximation of TV distance, a weaker guarantee
than relative accuracy when the distance may be small.  A separate line of
work~\cite{devroyeMehrabianReddad2018gaussians,
arbasAshtianiLiaw2023learning,kontorovich2024tensorization,
avitalKontorovich2024sharp} obtained explicit estimates with fixed
relative-accuracy guarantees for structured high-dimensional distributions.

\paragraph{Partition-function approximation.}
Our reduction uses deterministic partition-function
approximation.  For the hard-core model,
Weitz~\cite{weitz2006counting} introduced the self-avoiding-walk tree and
correlation-decay method and obtained an FPTAS for uniform activity below
$\lambda_{\mathrm c}(\Delta)$.  Harvey et
al.~\cite{harveySrivastavaVondrak2018independence} gave a deterministic
algorithm for complex external fields lying a fixed multiplicative distance
inside the Shearer region; in particular, it applies when
$\abs{z_v}\le(1-\eta)\lambda_{\mathrm{s}}(\Delta)$ for every vertex and a fixed
$\eta>0$.  However, Bez{\'a}kov{\'a} et
al.~\cite{bezakova2018inapproximability} showed that, for some complex
activities with modulus between $\lambda_{\mathrm{s}}(\Delta)$ and
$\lambda_{\mathrm{c}}(\Delta)$, approximating the hard-core partition function
is $\#\mathrm{P}$-hard.  More generally, Patel and
Regts~\cite{patelRegts2017deterministic} showed how zero-free regions for a
broad class of bounded-degree graph polynomials lead to deterministic
approximation algorithms.  Liu et
al.~\cite{liuSinclairSrivastava2019correlation} further connected real
correlation decay with zero-freeness in scalar complex neighborhoods; the
nonuniform complex approximation used here is obtained by the coordinatewise
extensions developed in later sections.

For the Ising model, Jerrum and Sinclair~\cite{jerrumSinclair1993ising} gave
the classical FPRAS for the ferromagnetic partition function.  Sinclair et
al.~\cite{sinclairSrivastavaThurley2012approximation} and Li et
al.~\cite{li2013correlation} developed deterministic correlation-decay
algorithms for antiferromagnetic two-spin systems throughout the interior of
the tree-uniqueness region.  In the ferromagnetic setting, Liu et
al.~\cite{liuSinclairSrivastava2017ising} used Lee--Yang zero-freeness to
obtain a deterministic FPTAS on bounded-degree graphs away from zero field.

\section{Technical overview}
\label{sec:technical-overview}

We use the hard-core model to illustrate the common reduction.  Fix
$0<b<r<\infty$, let $G=(V,E)$ have maximum degree at most $\Delta$, write
$n=\abs V$, and let $\lambda^\mu,\lambda^\nu\in[b,r]^V$ be two distinct
activity vectors.  Let $\mu=\mu_{G,\lambda^\mu}$ and
$\nu=\mu_{G,\lambda^\nu}$, and write $\Omega$ for their common feasible
support.  We abbreviate $Z(\lambda)=Z_G(\lambda)$, interpret vector
logarithms coordinatewise, and, for $x\in\R$, write
$x_+=\max\{x,0\}$.  The total variation distance can be expressed as
\begin{equation}
\begin{aligned}
  \TV(\mu,\nu)
  &=\sum_{\substack{\sigma\in\Omega\\ \mu(\sigma)\ge\nu(\sigma)}}
      \bigl(\mu(\sigma)-\nu(\sigma)\bigr)\\
  &=\E_{\sigma\sim\mu}
      \left(1-\frac{\nu(\sigma)}{\mu(\sigma)}\right)_+.
\end{aligned}
\label{eq:tv-distance}
\end{equation}
Define
\[
  s=\norm{\log\lambda^\nu-\log\lambda^\mu}_\infty>0,
  \qquad
  a_v=\frac{\log\lambda_v^\nu-\log\lambda_v^\mu}{s},
  \qquad
  A(\sigma)=\sum_{v:\,\sigma_v=+1}a_v,
\]
and
\[
  c=\frac{\log Z(\lambda^\nu)-\log Z(\lambda^\mu)}{s}.
\]
Thus $\abs{a_v}\le1$ and $A(\sigma)\in[-n,n]$.  We can also prove that $c$ lies
in $[-n,n]$.  For every
$\sigma\in\Omega$, the likelihood ratio is
\begin{equation}
\begin{aligned}
  \frac{\nu(\sigma)}{\mu(\sigma)}
  &=\frac{w_{\lambda^\nu}(\sigma)/Z(\lambda^\nu)}
         {w_{\lambda^\mu}(\sigma)/Z(\lambda^\mu)}
   =\frac{Z(\lambda^\mu)}{Z(\lambda^\nu)}
     \prod_{v:\,\sigma_v=+1}\frac{\lambda_v^\nu}{\lambda_v^\mu}\\
  &=\e^{s(A(\sigma)-c)}.
\end{aligned}
\label{eq:tech-overview-ratio}
\end{equation}
For $s>0$, set
\[
  g_s(x)=\frac{(1-\e^{sx})_+}{s}.
\]
Combining \eqref{eq:tv-distance} and
\eqref{eq:tech-overview-ratio} gives the normalized identity
\begin{align}\label{eq:tech-overview-normalized-tv}
  \frac{\TV(\mu,\nu)}s
  =\E_{\sigma\sim\mu}g_s(A(\sigma)-c).
\end{align}
This normalization is important when the two activity vectors are very
close.  The function $g_s$ is $1$-Lipschitz, uniformly in $s$, and
$0\le g_s(x)\le(-x)_+$.

\paragraph{The common deterministic reduction.}
There are two quantities to estimate in
\eqref{eq:tech-overview-normalized-tv}: the unknown center $c$ and the
expectation $\E_\mu[g_s(A-c)]$.  Consider the
coordinatewise geometric path
\[
  \lambda_v(t)
  =(\lambda_v^\mu)^{1-t}(\lambda_v^\nu)^t
  =\lambda_v^\mu\e^{tsa_v},
  \qquad t\in[0,1],
\]
and let $\pi_t=\mu_{G,\lambda(t)}$.  Differentiating the log partition
function along this path gives
\begin{equation}
  c=\int_0^1\E_{\sigma\sim\pi_t}A(\sigma)\,\mathrm dt
   =\int_0^1\sum_{v\in V}a_v
      \Pr_{\sigma\sim\pi_t}(\sigma_v=+1)\,\mathrm dt.
  \label{eq:tech-overview-thermodynamic}
\end{equation}
If $m(t)=\E_{\pi_t}A$, then we can prove that
\[
  m'(t)=s\operatorname{Var}_{\pi_t}(A),
  \qquad
  \abs{m'(t)}\le \log(r/b)n^2.
\]
Consequently, if we can deterministically approximate the marginal
probabilities $\Pr_{\sigma\sim\pi_t}(\sigma_v=+1)$ for all $v\in V$ and
$t\in[0,1]$ to additive error $\delta$, then we can choose a polynomial number
of uniformly spaced points $t$ in $[0,1]$ and take the average of the
corresponding $m(t)$-values.  This gives a deterministic additive approximation
$\widehat c$ to $c$ with additive error $O(\delta n)$.  Projecting
$\widehat c$ onto $[-n,n]$ cannot increase this error.

To implement the marginal oracle, we use deterministic partition-function
approximation for positive real external fields.  Correlation-decay algorithms
provide the required relative approximations when all activities lie
a fixed distance below the uniqueness threshold
\cite{weitz2006counting,li2013correlation}.

To estimate the remaining expectation deterministically, apply the
constructive Jackson approximation from \cite{devoreLorentz1993constructive} to
\[
  F(x)=g_s(x-\widehat c),\qquad x\in[-n,n].
\]
For any desired additive accuracy $\delta\in(0,1)$, we construct a trigonometric
polynomial
\[
  P(x)=\sum_{\abs k\le K}\alpha_k\e^{\iu t_kx},
  \qquad K=O(n/\delta),\qquad t_k=\frac{\pi k}{2n},
\]
which uniformly approximates $F$ and whose coefficients have $\ell_1$-norm
polynomially bounded in $n$ and $1/\delta$.  We also compute an estimate
$\widehat\alpha_k$ of each $\alpha_k$ for $\abs{k}\leq K$.  Recall that
$\Phi_{\mu,A}(t)=\E_{\sigma\sim\mu}\e^{\iu tA(\sigma)}$ is the
characteristic function of $A$ under $\mu$.  Expanding $P$ and interchanging
the two finite sums gives
\begin{align}
  \E_{\sigma\sim\mu}P(A(\sigma))
  &=\sum_{\sigma\in\Omega}\mu(\sigma)
      \sum_{\abs k\le K}\alpha_k\e^{\iu t_kA(\sigma)}
      \notag\\
  &=\sum_{\abs k\le K}\alpha_k
      \sum_{\sigma\in\Omega}\mu(\sigma)\e^{\iu t_kA(\sigma)}
      \notag\\
  &=\sum_{\abs k\le K}\alpha_k\Phi_{\mu,A}(t_k).
  \label{eq:tech-overview-fourier-expectation}
\end{align}
The required characteristic-function values are themselves ratios of
partition functions:
\begin{equation}
  \Phi_{\mu_{G,\lambda},A}(t)
  =
  \frac{Z_G((\lambda_v\e^{\iu ta_v})_{v\in V})}
       {Z_G(\lambda)}.
  \label{eq:tech-overview-characteristic-ratio}
\end{equation}
For the Ising model, the same identity holds with the common edge activity
$\beta$ fixed in the numerator and denominator.  The error in estimating the
expectation in~\eqref{eq:tech-overview-normalized-tv} comes from four sources:
\begin{enumerate}
  \item The $O(\delta n)$ additive error in the estimate $\widehat c$ of $c$.
  \item The additive error between $F(x)$ and its trigonometric approximation $P(x)$ for all possible $x$, which is guaranteed by the constructive Jackson approximation.
  \item The additive error in estimating the coefficients $\alpha_k$ by $\widehat\alpha_k$, which is also guaranteed by the constructive Jackson approximation.
  \item The additive error in estimating the characteristic-function values.\label{error-4}
\end{enumerate}

To convert this additive guarantee into a relative one, in each of
the parameter regimes in our theorems, our parameter-separation bounds show
that
\[
  \TV(\mu,\nu)\ge\kappa s
\]
for a positive constant $\kappa$ depending only on the fixed model
parameters.  For the hard-core model one may take $\kappa=b/5000$; for the
Ising model, uniform lower bounds on conditional spin probabilities give
such a constant depending on $\Delta,\beta,b,r$.  Hence the normalized
target
\[
  q=\frac{\TV(\mu,\nu)}s
\]
is at least $\kappa$, so additive accuracy $O(\varepsilon\kappa)$ suffices
for an $\varepsilon$-relative approximation.  Identical inputs are detected
exactly and mapped to zero; otherwise we return
$\widehat d=s\max\{0,\widehat q\}$.  It remains only to show how to implement
the characteristic-function oracle in \cref{error-4}.

\paragraph{Regimes with a global complex oracle.}
To implement the characteristic-function oracle, we need to deterministically
approximate the partition function at complex external fields.
In the strict Shearer regime for the hard-core model, the known complex
partition-function algorithm in \cite{harveySrivastavaVondrak2018independence}
applies to every $z\in\C^V$ satisfying $\abs{z_v}\le r$ for all $v\in V$, where
$r=(1-\eta)\lambda_{\mathrm{s}}(\Delta)$.
$\lambda_{\mathrm{s}}(\Delta)=\frac{(\Delta-1)^{\Delta-1}}{\Delta^\Delta}$ is the Shearer threshold.
Here $\eta>0$ is fixed, so the running-time dependence on the distance from
the boundary is absorbed into the model-parameter dependence.
Therefore, applying the algorithm directly to the complex vector gives a
deterministic characteristic-function oracle.
For the ferromagnetic Ising model, the multivariate Lee--Yang theorem gives
zero-freeness whenever all complex activities have modulus at most a fixed
$r<1$.  The coefficient method of
\cite{liuSinclairSrivastava2017ising} then supplies the same type of oracle;
\cref{sec:lee-yang-regime} gives the details.

\paragraph{From a local complex oracle to arbitrary phases.}
The main difficulty for the hard-core uniqueness regime is that the global
polydisc oracle is restricted to moduli strictly below
$\lambda_{\mathrm s}(\Delta)$, whereas
\[
  \frac{\lambda_{\mathrm s}(\Delta)}
       {\lambda_{\mathrm c}(\Delta)}
  =\left(\frac{\Delta-2}{\Delta}\right)^\Delta<1.
\]
It therefore does not in general cover the whole interval
$[b,(1-\eta)\lambda_{\mathrm c}(\Delta)]$.  In fact, Bez{\'a}kov{\'a} et
al.~\cite{bezakova2018inapproximability} showed that, for some complex
activities with modulus between $\lambda_{\mathrm{s}}(\Delta)$ and
$\lambda_{\mathrm{c}}(\Delta)$, approximating the hard-core partition function
is $\#\mathrm{P}$-hard.

We first obtain a
\emph{local} complex oracle: there is a constant $\theta_0>0$ such that, if
\[
  z_v=\lambda_v\e^{\iu\theta_v},
  \qquad
  \lambda_v\in[b,(1-\eta)\lambda_{\mathrm c}(\Delta)],
  \qquad
  \abs{\theta_v}\le\theta_0,
\]
then the partition function is zero-free throughout this region and can be
deterministically approximated to any prescribed relative error, uniformly
over induced subgraphs.  To prove this, we extend a result of Liu et
al.~\cite{liuSinclairSrivastava2019correlation} to nonuniform external fields.
We compute occupation ratios on the pruned Weitz self-avoiding-walk
tree~\cite{weitz2006counting}.  For a tree vertex $x$, $\pi(x)$ is the
endpoint in $G$ of the walk represented by $x$.  Let $Z_x^{(s)}(z)$ be the
partition function of the subtree rooted at $x$, with inherited boundary
conditions and $x$ fixed to be occupied ($s=1$) or unoccupied ($s=0$).
The occupation ratio and its update map are
\[
  R_x:=\frac{Z_x^{(1)}(z)}{Z_x^{(0)}(z)},
  \qquad F_{z,d}(\mathbf R):=z\prod_{j=1}^d(1+R_j)^{-1}.
\]
If $x$ is unpinned with children $x_1,\ldots,x_d$, then
$R_x=F_{z_{\pi(x)},d}(R_{x_1},\ldots,R_{x_d})$.

To control how errors propagate, set $y_x=\varphi(R_x)$, where
$\varphi(t)=2\operatorname{arcsinh}\sqrt t$ and
$\varphi^{-1}(y)=\sinh^2(y/2)$.  The update becomes
\[
  F_{z,d}^{\varphi}(\mathbf y)
  :=\varphi\!\left(
    F_{z,d}(\varphi^{-1}(y_1),\ldots,\varphi^{-1}(y_d))
  \right).
\]
Put $r=(1-\eta)\lambda_{\mathrm c}(\Delta)$ and
$\mathcal I=[0,\varphi(r)]$.  By the real contraction of Li et
al.~\cite{li2013correlation} and compactness, there is a uniform $\eta_0>0$
such that
\[
  \abs{F_{a,d}^{\varphi}(\mathbf y)-F_{a,d}^{\varphi}(\mathbf y')}
  \le(1-\eta_0)\norm{\mathbf y-\mathbf y'}_\infty
\]
for $a\in[b,r]$, $\mathbf y,\mathbf y'\in\mathcal I^d$, and $d\le\Delta-1$.
Thus, for a fixed activity, each non-root update reduces the largest error
in its child values.

Following Liu et al.~\cite{liuSinclairSrivastava2019correlation},
analyticity and compactness give a common complex neighborhood of $[b,r]$
where contraction holds even when activities differ between vertices.
For sufficiently small phases, the root perturbation bound gives
$\abs{1+R_{\mathrm{root}}}\ge1/2$ uniformly over induced subgraphs.
Induction using the deletion identity then proves zero-freeness.
Truncating the tree and reconstructing the partition function from these
ratios yields deterministic relative approximation.

The characteristic-function queries, however, may contain arbitrary
phases.  Reduce a given phase vector $\theta$ coordinatewise modulo $2\pi$
to $\overline\theta\in[-\pi,\pi]^V$.  This reduction leaves the
characteristic function unchanged.  Set
\[
  z_v=\lambda_v\e^{\iu\overline\theta_v}
  \qquad(v\in V),
\]
and define
\[
  W=\set{v\in V:\abs{\overline\theta_v}\ge\theta_0/2}.
\]
Our phase-damping lemma (\cref{lem:phase-damping}) gives a constant
$\gamma>0$, depending only on the fixed model parameters, such that
\begin{equation}
  \abs{
    \E_{\sigma\sim\mu_{G,\lambda}}
      \exp\left(\iu\sum_{v\in V}
        \overline\theta_v\one_{\{\sigma_v=+1\}}\right)
  }
  \le\e^{-\gamma\abs W}.
  \label{eq:tech-overview-phase-damping}
\end{equation}
The intuition is that many large phases force cancellation.  We select a
constant fraction of the vertices in $W$ whose neighborhoods do not
interact.  After conditioning on all spins outside these neighborhoods, the
remaining blocks are independent.  In each block, two configurations each
have probability at least a fixed positive constant.  Their complex
contributions point in directions separated by at least a fixed positive
angle, so they partially cancel.  Repeating this over
all selected blocks gives the exponential decay in
\eqref{eq:tech-overview-phase-damping}.

To approximate the characteristic function to additive error $\delta$, we
consider two cases.  If $\abs W>\gamma^{-1}\log(8/\delta)$, then
\eqref{eq:tech-overview-characteristic-ratio} and
\eqref{eq:tech-overview-phase-damping} show that the numerator satisfies
\[
  \abs{Z_G(z)}
  <\frac{\delta}{8}Z_G(\lambda).
\]
Hence the characteristic function is close to zero, so we simply return zero.
Otherwise, $\abs W=O(\log(1/\delta))$, and we approximate the numerator
$Z_G(z)$ using the following decomposition.  For each
$J\in\cI(G[W])$, set $H_J=G-(W\cup N_G(J))$.  Then
\begin{equation}
  Z_G(z)
  =\sum_{J\in\cI(G[W])}z^J
     Z_{H_J}(z|_{V(H_J)}).
  \label{eq:tech-overview-large-phase-decomposition}
\end{equation}
Here $\cI(G[W])$ is the family of independent subsets of $W$,
$N_G(J)$ is the open neighborhood of $J$, and
$z^J=\prod_{v\in J}z_v$.  All fields in the residual partition functions
have phase magnitude less than $\theta_0/2$, so the local oracle applies to
every term.  Cancellation does not amplify the accumulated error, because
positivity at the corresponding real activities gives
\[
  \sum_{J\in\cI(G[W])}
    \abs{z^JZ_{H_J}(z|_{V(H_J)})}
  \le Z_G(\lambda).
\]
There are at most
\[
  2^{\abs W}
  \le(8/\delta)^{(\log 2)/\gamma}
  \le\poly_{\Delta,\eta,b}(1/\delta)
\]
terms.  Relative approximation of the residual terms yields a numerator error
of $O(\delta)Z_G(\lambda)$.  Dividing by a relative approximation to the
positive-real denominator therefore yields the required additive
all-frequency characteristic-function oracle.  This
damping-and-enumeration step is what allows us to use only local complex
approximation while still reaching the full constant-gap uniqueness regime,
proving \cref{thm:uniqueness-main}.

\section{Preliminaries}
\label{sec:preliminaries}

Let $G=(V,E)$ be a finite simple graph, and write $n=\abs V$.  The family of
independent sets of $G$ is denoted by $\cI(G)$.  For $U,S\subseteq V$, let
$G[U]$ be the induced subgraph and write $G-S=G[V\setminus S]$ (in
particular, $G-u=G-\{u\}$).  We write $N_G(u)$ for the open neighborhood of
$u$, $N_G[u]=N_G(u)\cup\{u\}$, and
$N_G(S)=\bigcup_{u\in S}N_G(u)$.  The edge boundary of $S$ is
\[
  \delta_G(S)
  =\set{\{u,v\}\in E:\abs{\{u,v\}\cap S}=1}.
\]
For a complex vector $z=(z_v)_{v\in V}\in\C^V$ and $S\subseteq V$, write
$z^S=\prod_{v\in S}z_v$.  The \emph{independence polynomial} is
\begin{equation*}
  Z_G(z)
  =\sum_{I\in\cI(G)}z^I.
\end{equation*}

For two probability distributions $\mu$ and $\nu$ on the same finite space
$\Omega$, their \emph{total variation distance} is
\begin{equation*}
  \TV(\mu,\nu)
  =\frac12\sum_{x\in\Omega}\abs{\mu(x)-\nu(x)}
  =\sup_{\mathcal A\subseteq\Omega}
    \abs{\mu(\mathcal A)-\nu(\mathcal A)}.
\end{equation*}

For a probability distribution $\mu$ on $\Omega$ and a function
$A:\Omega\to\R$, the \emph{characteristic function} of $A$ under $\mu$ is
\[
  \Phi_{\mu,A}(t)
  =\E_{X\sim\mu}\e^{\iu tA(X)},
  \qquad t\in\R.
\]

For $D\subseteq\R$ and $\ell\geq0$, a function $f:D\to\C$ is
\emph{$\ell$-Lipschitz} if
\[
  \abs{f(x)-f(y)}\leq\ell\abs{x-y}
  \qquad(x,y\in D).
\]

Throughout, we use the unit-cost real/complex arithmetic-oracle model:
arithmetic and elementary-function evaluations have unit cost, and oracle
calls are charged their stated running times.  We do not consider bit
complexity.

\section{The deterministic framework}
\label{sec:common-framework}

In this section, we isolate the part of the argument shared by the hard-core
and Ising models.  For the empty graph, both Gibbs distributions are
concentrated on the unique empty configuration, so the algorithm returns
zero in constant time.  Henceforth in this section, assume $n=\abs V\ge1$.
Throughout
\cref{sec:fourier-expectation,sec:normalization,sec:uniform-separation,%
sec:oracle-reduction}, fix a graph $G=(V,E)$ and,
for the Ising model, also fix the edge activity $\beta$.  Write $\Omega$ for
the common support of the two measures: the feasible configurations for the
hard-core model and $\{-1,+1\}^V$ for the Ising model.  According to the model
under consideration, write
\[
  w_\lambda(\sigma),\qquad Z(\lambda),\qquad
  \pi_\lambda(\sigma)=\frac{w_\lambda(\sigma)}{Z(\lambda)}
\]
for $w_\lambda,Z_G,\mu_{G,\lambda}$ or for
$w_{\beta,\lambda},Z_G(\beta,\cdot),\mu_{G,\beta,\lambda}$,
respectively.  In either case, changing only the external fields gives
\begin{equation}
  \frac{w_{\lambda'}(\sigma)}{w_\lambda(\sigma)}
  =\prod_{v:\,\sigma_v=+1}\frac{\lambda'_v}{\lambda_v}.
  \label{eq:common-weight-ratio}
\end{equation}
No additional representation of a configuration will be needed.

\subsection{A deterministic Fourier expectation oracle}
\label{sec:fourier-expectation}

\begin{lemma}
\label{lem:quotient-estimate}
Let $y>0$, $x\in\C$, and $\abs{x}\le y$.  If
$\varepsilon\in(0,1)$ and
\[
  \abs{\widehat x-x}\le\varepsilon y,
  \qquad
  \abs{\widehat y-y}\le\varepsilon y,
\]
then
\[
  \abs{\frac{\widehat x}{\widehat y}-\frac{x}{y}}
  \le\frac{2\varepsilon}{1-\varepsilon}.
\]
\end{lemma}

\begin{proof}
The lower bound $\abs{\widehat y}\ge(1-\varepsilon)y$ and the triangle inequality give
\[
  \abs{\frac{\widehat x}{\widehat y}-\frac{x}{y}}
  \le
  \frac{\abs{\widehat x-x}}{\abs{\widehat y}}
  +\frac{\abs x\,\abs{\widehat y-y}}
        {y\abs{\widehat y}}
  \le\frac{2\varepsilon}{1-\varepsilon}.
\]
\end{proof}

For $\lambda\in\R_{>0}^V$ and real coefficients $(a_v)_{v\in V}$, define
\[
  A(\sigma)=\sum_{v:\,\sigma_v=+1}a_v.
\]
For every $t\in\R$, expanding the partition function gives the common
characteristic-function identity
\begin{equation}
  \Phi_{\pi_\lambda,A}(t)
  =
  \frac{Z((\lambda_v\e^{\iu ta_v})_{v\in V})}
       {Z(\lambda)}.
  \label{eq:characteristic-partition-ratio}
\end{equation}

Following the Jackson-kernel construction in
\cite[Chapter~7, Section~2]{devoreLorentz1993constructive}, we give and
analyze a deterministic algorithm for computing the coefficients needed in
our application.

\begin{lemma}[Constructive Jackson approximation]
\label{lem:constructive-jackson}
Let $n\geq1$, let $\varepsilon\in(0,1)$, and let
$F:[-n,n]\to[0,2n]$ be $1$-Lipschitz.  Suppose that $F$ can be evaluated exactly in time $T_F(n)$.  Set
\[
  K=\left\lceil\frac{24n}{\varepsilon}\right\rceil,
  \qquad
  t_k=\frac{\pi k}{2n},
  \qquad \abs{k}\leq K.
\]
There are coefficients $\alpha_k$ for which
\[
  P(x)=\sum_{\abs{k}\leq K}\alpha_k\e^{\iu t_kx}
\]
is real-valued and satisfies
\begin{equation}
  \sup_{x\in[-n,n]}\abs{P(x)-F(x)}\leq\frac{\varepsilon}{4},
  \qquad
  \sum_{\abs{k}\leq K}\abs{\alpha_k}\leq6nK.
  \label{eq:jackson-guarantees}
\end{equation}
Moreover, one can deterministically compute coefficients
$\widehat\alpha_k$ such that
\begin{equation}
  \sum_{\abs{k}\leq K}
    \abs{\widehat\alpha_k-\alpha_k}
  \leq\frac{\varepsilon}{8}.
  \label{eq:jackson-coefficient-error}
\end{equation}
This takes
\[
  O\left(
    \frac{n^3}{\varepsilon^3}
    \left[T_F(n)+\log\!\left(\frac{2n}{\varepsilon}\right)\right]
  \right)
\]
time.  In particular, the running time is
$\poly(n,1/\varepsilon)$ whenever $T_F(n)=\poly(n)$.
\end{lemma}

The proof is deferred to Appendix~\ref{app:constructive-jackson}.

The next lemma is the deterministic replacement for sampling.

\begin{lemma}
\label{lem:fourier-expectation}
Let $n=\abs V\ge1$, let $\lambda\in\R_{>0}^V$, write
$\pi=\pi_\lambda$, and let
$A(\sigma)=\sum_{v:\,\sigma_v=+1}a_v$ with $\abs{a_v}\le1$.  Suppose that, for every
$t\in\R$ and $\varepsilon'\in(0,1)$, a deterministic
characteristic-function oracle returns $\widehat\Phi_{\pi,A}(t)$ satisfying
\[
  \abs{\widehat\Phi_{\pi,A}(t)-\Phi_{\pi,A}(t)}
  \le\varepsilon'
\]
in time $T_{\widehat\phi}(n,\varepsilon')$.  Let $F:[-n,n]\to[0,2n]$ be
$1$-Lipschitz and exactly evaluable in time $T_F(n)$.  Then, for every
$\varepsilon\in(0,1)$, one can deterministically compute
$\widehat M$ such that
\[
  \abs{\widehat M-\E_{\sigma\sim\pi}F(A(\sigma))}
  \le\varepsilon
\]
in time
\[
  \widetilde O\left(
    \frac{n^3}{\varepsilon^3}T_F(n)
    +\frac{n}{\varepsilon}T_{\widehat\phi}\left(
      n,\frac{\varepsilon}{96n\lceil24n/\varepsilon\rceil+2\varepsilon}
    \right)
  \right).
\]
Here $\widetilde O$ hides logarithmic factors in $n$ and $1/\varepsilon$.
\end{lemma}

\begin{proof}
Since $\abs{A(\sigma)}\le n$, we have $A(\sigma)\in[-n,n]$.  Apply
\cref{lem:constructive-jackson} with accuracy
$\varepsilon$, and let $K$, $t_k$, $\alpha_k$, and $\widehat\alpha_k$ be as in
that lemma.  By \eqref{eq:jackson-guarantees} and
\eqref{eq:jackson-coefficient-error},
\[
  S_{\widehat\alpha}
  =\sum_{\abs k\le K}\abs{\widehat\alpha_k}
  \le 6nK+\varepsilon/8.
\]
Let $\varepsilon'
  =\varepsilon/[16(6nK+\varepsilon/8)]$.
For every $\abs k\le K$, invoke the characteristic-function oracle with
inputs $t_k$ and $\varepsilon'$.  Define
\[
  \widehat M
  =\Re\sum_{\abs k\le K}
       \widehat\alpha_k\widehat\Phi_{\pi,A}(t_k).
\]

We have
\begin{align*}
  \sum_{\abs k\le K} \alpha_k\Phi_{\pi,A}(t_k)
  &=\sum_{\abs k\le K} \alpha_k\sum_{\sigma\in\Omega}\pi(\sigma)
      \e^{\iu t_k A(\sigma)}
   =\sum_{\sigma\in\Omega}\pi(\sigma)
      \sum_{\abs k\le K}\alpha_k\e^{\iu t_k A(\sigma)}\\
  &=\sum_{\sigma\in\Omega}\pi(\sigma)P(A(\sigma))
   =\E_{\sigma\sim\pi}P(A(\sigma)).
\end{align*}
Since the target is real, taking the real part cannot increase the error.
Therefore, we have
\begin{align*}
  &\abs{\widehat M-\E_{\sigma\sim\pi}F(A(\sigma))}\\
  \leq& \abs{\sum_{\abs k\le K}
                \widehat\alpha_k \widehat\Phi_{\pi,A}(t_k)
              -\sum_{\abs k\le K}
                \widehat\alpha_k \Phi_{\pi,A}(t_k)}\\
  +&\abs{\sum_{\abs k\le K}
                \widehat\alpha_k \Phi_{\pi,A}(t_k)
              -\sum_{\abs k\le K} \alpha_k\Phi_{\pi,A}(t_k)}
     +\abs{\sum_{\abs k\le K} \alpha_k\Phi_{\pi,A}(t_k)
        -\E_{\sigma\sim\pi}F(A(\sigma))}\\
  \leq& \sum_{\abs k\le K}
        \abs{\widehat\alpha_k}
        \abs{\widehat\Phi_{\pi,A}(t_k)-\Phi_{\pi,A}(t_k)}
     +\sum_{\abs k\le K}
        \abs{\widehat\alpha_k-\alpha_k}
        \abs{\Phi_{\pi,A}(t_k)}
     +\abs{\E_{\sigma\sim\pi}\{P(A(\sigma))-F(A(\sigma))\}}\\
  \leq& \varepsilon' S_{\widehat\alpha}
     +\frac{\varepsilon}{8}
     +\frac{\varepsilon}{4}\\
  \leq& \frac{\varepsilon}{16}
     +\frac{\varepsilon}{8}
     +\frac{\varepsilon}{4}<\varepsilon,
\end{align*}
where we used the fact that $\abs{\Phi_{\pi,A}(t_k)}\leq 1$ for all $k$.
By \cref{lem:constructive-jackson}, computing the coefficients takes
$\widetilde O(n^3T_F(n)/\varepsilon^3)$ time.  Since
$K=\lceil24n/\varepsilon\rceil$, there are
$2\lceil24n/\varepsilon\rceil+1=O(n/\varepsilon)$ oracle calls, each taking
time $T_{\widehat\phi}(n,\varepsilon')$, and the remaining computation takes
$O(n/\varepsilon)$ time.  The total running time is therefore
\[
  \widetilde O\left(
    \frac{n^3}{\varepsilon^3}T_F(n)
    +\frac{n}{\varepsilon}T_{\widehat\phi}\left(
      n,\frac{\varepsilon}{96n\lceil24n/\varepsilon\rceil+2\varepsilon}
    \right)
  \right).
\]
\end{proof}

\subsection{Normalization and thermodynamic integration}
\label{sec:normalization}

Throughout this subsection, let $n=\abs V$, let $0<b<r<\infty$, and let
$\lambda^\mu,\lambda^\nu\in[b,r]^V$.  Define
$\mu=\pi_{\lambda^\mu}$ and $\nu=\pi_{\lambda^\nu}$.  If the two vectors
agree, then $\mu=\nu$ and
the algorithm returns $0$.  Hence assume they differ and define
\[
  s=\norm{\log\lambda^\nu-
                 \log\lambda^\mu}_\infty>0,
  \qquad
  a_v=\frac{\log\lambda_v^\nu-\log\lambda_v^\mu}{s},
  \qquad
  A(\sigma)=\sum_{v:\,\sigma_v=+1}a_v.
\]
Thus $\abs{a_v}\le1$, $A(\sigma)\in[-n,n]$, and
\begin{equation}
  \frac{w_{\lambda^\nu}(\sigma)}
       {w_{\lambda^\mu}(\sigma)}
  =\e^{sA(\sigma)}.
  \label{eq:weight-ratio}
\end{equation}
Let
\begin{equation}
  c=\frac{\log Z(\lambda^\nu)-\log Z(\lambda^\mu)}{s}.
  \label{eq:normalized-center}
\end{equation}
Because
\[
  \frac{Z(\lambda^\nu)}{Z(\lambda^\mu)}
  =\E_{\sigma\sim\mu}\e^{sA(\sigma)},
\]
$\e^{sc}$ lies between $\e^{s\min_{\sigma\in\Omega}A(\sigma)}$ and
$\e^{s\max_{\sigma\in\Omega}A(\sigma)}$.  Consequently,
\begin{equation}
  c\in
  [\min_{\sigma\in\Omega}A(\sigma),
   \max_{\sigma\in\Omega}A(\sigma)]
  \subseteq[-n,n].
  \label{eq:center-range}
\end{equation}
Directly subtracting two approximate log partition functions would require
precision proportional to $s$.  The following lemma avoids this problem by
giving a procedure for estimating $c$.

\begin{lemma}
\label{lem:thermodynamic-identity}
For $t\in[0,1]$, define the coordinatewise geometric path
\[
  \begin{aligned}
    \lambda_v(t)
    &= (\lambda_v^\mu)^{1-t}(\lambda_v^\nu)^t
      =\lambda_v^\mu\e^{tsa_v},\\
    \lambda(t)&=(\lambda_v(t))_{v\in V},
    &\pi_t&=\pi_{\lambda(t)}.
  \end{aligned}
\]
Then
\begin{equation}
  c=\int_0^1 \E_{\sigma\sim\pi_t}A(\sigma)\,\mathrm dt.
  \label{eq:thermodynamic-integral}
\end{equation}
Moreover, if $m(t)=\E_{\sigma\sim\pi_t}A(\sigma)$, then
\begin{align}\label{lip-of-m}
    m'(t)=s\operatorname{Var}_{\pi_t}(A),
  \qquad
  \abs{m'(t)}\le \log(r/b)n^2.
\end{align}

Suppose in addition that $G$ has maximum degree at most $\Delta$ and that,
for every $\lambda'\in[b,r]^V$, $v\in V$, and
$\varepsilon'\in(0,1)$, a deterministic marginal-probability oracle returns
$\widehat p_v(\lambda')$ satisfying
\[
  \abs{
    \widehat p_v(\lambda')
    -\Pr_{\sigma\sim\pi_{\lambda'}}(\sigma_v=+1)
  }
  \le\varepsilon'
\]
in time $T_{\widehat p}(n,\varepsilon')$, with constants depending only on the
fixed model parameters.  Then, for every
$\varepsilon\in(0,1)$, one can deterministically compute
$\widehat c\in[-n,n]$ such that
\[
  \abs{\widehat c-c}\le\varepsilon
\]
in time
\[
  O\left(\frac{n^3}{\varepsilon}
    T_{\widehat p}\left(n,\frac{\varepsilon}{8n}\right)\right),
\]
with the same parameter dependence.
\end{lemma}

\begin{proof}
Let $Z(t)=Z(\lambda(t))$.  Differentiating its finite sum gives
\[
  \frac{\mathrm d}{\mathrm dt}\log Z(t)
  =s\E_{\sigma\sim\pi_t}A(\sigma).
\]
Integration from $0$ to $1$, followed by division by $s$, proves
\eqref{eq:thermodynamic-integral}.  To see the second-derivative identity,
define
\[
  \pi_t(\sigma)
  =\frac{w_{\lambda^\mu}(\sigma)\e^{tsA(\sigma)}}{Z(t)},
  \qquad
  N(t)=\sum_{\sigma\in\Omega}
    A(\sigma)w_{\lambda^\mu}(\sigma)\e^{tsA(\sigma)}.
\]
Then $m(t)=N(t)/Z(t)$.  Since these are finite sums,
termwise differentiation gives
\[
  N'(t)=sZ(t)\E_{\sigma\sim\pi_t}A(\sigma)^2,
  \qquad
  Z'(t)=sZ(t)\E_{\sigma\sim\pi_t}A(\sigma).
\]
The quotient rule now yields
\[
  m'(t)
  =\frac{N'(t)Z(t)-N(t)Z'(t)}{Z(t)^2}
  =s\left(\E_{\sigma\sim\pi_t}A(\sigma)^2
       -\bigl(\E_{\sigma\sim\pi_t}A(\sigma)\bigr)^2\right)
  =s\operatorname{Var}_{\pi_t}(A).
\]
Since $\abs{A(\sigma)}\le n$ for every $\sigma\in\Omega$,
$\operatorname{Var}_{\pi_t}(A)\le n^2$.  Finally,
$s\le\log(r/b)$, which proves the stated derivative bound.

For the algorithmic assertion, let
\[
  L_m=\log(r/b)n^2,
  \qquad
  M
  =\max\set{1,\left\lceil\frac{2L_m}{\varepsilon}\right\rceil},
  \qquad
  t_j=\frac{j}{M},
\]
for $0\le j<M$.  Let $\varepsilon'=\varepsilon/(8n)$.  For
$0\le j<M$ and $v\in V$, use the marginal-probability oracle
with inputs $\lambda(t_j)$, $v$, and $\varepsilon'$ to obtain
$\widehat p_{j,v}$ satisfying
\[
  \abs{
    \widehat p_{j,v}
    -\Pr_{\sigma\sim\pi_{t_j}}(\sigma_v=+1)
  }
  \le\varepsilon'.
\]
Define
\[
  \widehat m_j=\sum_{v\in V}a_v\widehat p_{j,v},
  \qquad
  \widehat c_0
  =\frac1{M}
     \sum_{j=0}^{M-1}\widehat m_j.
\]
The geometric path lies in $[b,r]^V$, so every oracle call is valid.  Since
$m$ is $L_m$-Lipschitz by~\eqref{lip-of-m}, we have
\[
  \abs{
    \int_0^1m(t)\,\mathrm dt
    -\frac1{M}
       \sum_{j=0}^{M-1}m(t_j)
  }
  \le\frac{L_m}{2M}
  \le\frac\varepsilon4.
\]
Moreover, $\abs{a_v}\le1$ gives
\[
  \abs{\widehat m_j-m(t_j)}
  \le n\varepsilon'
  =\frac\varepsilon8.
\]
Consequently,
\[
  \abs{\widehat c_0-c}
  \le\frac\varepsilon4+\frac\varepsilon8
  =\frac{3\varepsilon}{8}
  <\varepsilon.
\]
Projecting
$\widehat c_0$ onto $[-n,n]$ gives $\widehat c$ and cannot increase the
error by \eqref{eq:center-range}.

The algorithm makes $nM$ marginal-oracle calls, where
\[
  nM
  \le n+\frac{2\log(r/b)n^3}{\varepsilon}
  =O(n^3/\varepsilon).
\]
Each call takes time $T_{\widehat p}(n,\varepsilon/(8n))$.
Constructing the field vectors and computing
the weighted sums take another $O(nM)$ arithmetic operations.  The total
running time is therefore
\[
  O\left(\frac{n^3}{\varepsilon}
    T_{\widehat p}\left(n,\frac{\varepsilon}{8n}\right)\right),
\]
with the same parameter dependence.
\end{proof}

For $s>0$, define
\[
  g_s(x)=\frac{(1-\e^{sx})_+}{s}.
\]

\begin{lemma}
\label{lem:normalized-tv}
The function $g_s$ is $1$-Lipschitz on $\R$ and satisfies
$0\le g_s(x)\le(-x)_+$.  Furthermore,
\begin{equation}
  \frac{\TV(\mu,\nu)}s
  =\E_{\sigma\sim\mu}g_s\bigl(A(\sigma)-c\bigr).
  \label{eq:normalized-tv}
\end{equation}
The random argument $A(\sigma)-c$ lies in $[-2n,2n]$, and the integrand lies in
$[0,2n]$.
\end{lemma}

\begin{proof}
Both Gibbs distributions are positive on their common support $\Omega$, and hence
\[
  \TV(\mu,\nu)
  =\sum_{\substack{\sigma\in\Omega\\ \mu(\sigma)\ge\nu(\sigma)}}
      \bigl(\mu(\sigma)-\nu(\sigma)\bigr)
  =\E_{\sigma\sim\mu}
      \left(1-\frac{\nu(\sigma)}{\mu(\sigma)}\right)_+.
\]
By \eqref{eq:weight-ratio} and \eqref{eq:normalized-center},
\[
  \frac{\nu(\sigma)}{\mu(\sigma)}=\e^{s(A(\sigma)-c)}.
\]
Substitution proves \eqref{eq:normalized-tv}.  If $x<0$, then
$g_s'(x)=-\e^{sx}$, whose absolute
value is at most one; if $x>0$, then $g_s(x)=0$.  Continuity at zero proves the
Lipschitz claim.  The inequality $1-\e^{-u}\le u$ for $u\ge0$ gives
$g_s(x)\le(-x)_+$.  The range statements follow from
\eqref{eq:center-range} and $A(\sigma)\in[-n,n]$.
\end{proof}

\subsection{A uniform lower bound after normalization}
\label{sec:uniform-separation}

The normalized expectation in \eqref{eq:normalized-tv} must be bounded away
from zero.  The strict lower bound on the external fields supplies this fact.

\begin{lemma}
\label{lem:bound-dtv}
\label{lem:parameter-lower-bound}
Let $\mu$ and $\nu$ be the two hard-core measures defined above.
Suppose $\Delta\ge3$, $G$ has maximum degree at most $\Delta$, and all external
fields in both systems lie in $[b,r]$, where
$0<b<r<\lambda_{\mathrm c}(\Delta)$.  Then
\begin{equation}
  \TV(\mu,\nu)\ge\kappa_0 s,
  \qquad
  \kappa_0=\frac{b}{5000}.
  \label{eq:uniform-separation}
\end{equation}
Consequently, the expectation in \eqref{eq:normalized-tv} is at least
$\kappa_0$.
\end{lemma}

\begin{lemma}[Ising parameter separation]
\label{lem:ising-parameter-separation}
Fix an integer $\Delta\ge0$, $\beta>0$, and $0<b<r<\infty$.  Set
\[
  \hat\beta=\min\{\beta,\beta^{-1}\},
  \qquad
  \rho
  =\min\left\{
    \frac{b\,\hat\beta^\Delta}
         {1+b\,\hat\beta^\Delta},
    \frac{1}{1+r\,\hat\beta^{-\Delta}}
  \right\}.
\]
Suppose that $G$ has maximum degree at most $\Delta$ and that
$\lambda^\mu,\lambda^\nu\in[b,r]^V$.  For the two corresponding Ising
measures with the same edge activity $\beta$,
\begin{equation}
  \TV(\mu,\nu)\ge\kappa_{\mathrm I}s,
  \qquad
  \kappa_{\mathrm I}=\frac{\rho^2}{4(\Delta+1)}.
  \label{eq:ising-uniform-separation}
\end{equation}
\end{lemma}

The following is the specialization of
\cite[Lemma~15]{pmlr-v291-feng25a} to our parametrizations.  Indeed, up to a
configuration-independent factor, our Ising weights have
$J_{uv}=-\tfrac12\log\beta$ and $h_v=\tfrac12\log\lambda_v$; hence the
$\rho^2/2$ bound there becomes the $\rho^2/4$ bound below.

\begin{lemma}[\cite{pmlr-v291-feng25a}, Lemma~15]
\label{lem:dtv-citation}
\begin{enumerate}[label=\textnormal{(\roman*)},leftmargin=2.6em]
  \item Under the assumptions of \cref{lem:bound-dtv},
  \[
    \TV(\mu,\nu)
    \ge\frac{1}{5000}
      \max_{v\in V}\abs{\lambda_v^\mu-\lambda_v^\nu}.
  \]

  \item Let $\mu=\mu_{G,\beta,\lambda^\mu}$ and
  $\nu=\mu_{G,\beta,\lambda^\nu}$ be two Ising measures with the same edge
  activity $\beta>0$.  Suppose that, for some $\rho\in(0,1]$, after every
  partial pinning and for every vertex $v$ left unpinned, both measures
  satisfy
  \[
    \Pr(\sigma_v=\omega\mid\text{the pinning})\ge\rho
    \qquad(\omega\in\{-1,+1\}).
  \]
  Then
  \[
    \TV(\mu,\nu)
    \ge\frac{\rho^2}{4}
      \max_{v\in V}
      \frac{\abs{\log\lambda_v^\mu-\log\lambda_v^\nu}}
           {\deg_G(v)+1}.
  \]
\end{enumerate}
\end{lemma}

\begin{proof}[Proof of \cref{lem:bound-dtv}]
Choose $u\in V$ attaining the maximum in the definition of $s$ and assume,
without loss of generality, that $\lambda_u^\mu\ge\lambda_u^\nu$.  By
\cref{lem:dtv-citation}\textnormal{(i)},
\[
  \TV(\mu,\nu)
  \ge\frac{1}{5000}\abs{\lambda_u^\mu-\lambda_u^\nu}
  =\frac{1}{5000}\lambda_u^\nu
    \left(\frac{\lambda_u^\mu}{\lambda_u^\nu}-1\right)
  \ge\frac{b}{5000}(\e^s-1)
  \ge\frac{b}{5000}s.
\]
\end{proof}

\begin{proof}[Proof of \cref{lem:ising-parameter-separation}]
Fix $v\in V$ and a spin configuration $\tau$ on $N_G(v)$.  Under either
measure, the conditional probability of a positive spin at $v$ is
\[
  \Pr(\sigma_v=+1\mid\sigma_{N_G(v)}=\tau)
  =\frac{\lambda_v\beta^{m_-(\tau)-m_+(\tau)}}
         {1+\lambda_v\beta^{m_-(\tau)-m_+(\tau)}},
\]
where $m_+(\tau)$ and $m_-(\tau)$ are the numbers of positive and negative
spins in $\tau$.  Since
\[
  \abs{m_-(\tau)-m_+(\tau)}
  \le\deg_G(v)\le\Delta,
\]
the definition of $\hat\beta$ gives
\[
  \hat\beta^\Delta
  \le\beta^{m_-(\tau)-m_+(\tau)}
  \le\hat\beta^{-\Delta}.
\]
It follows that
\[
\begin{aligned}
  \Pr(\sigma_v=+1\mid\sigma_{N_G(v)}=\tau)
  &=\frac{\lambda_v\beta^{m_-(\tau)-m_+(\tau)}}
          {1+\lambda_v\beta^{m_-(\tau)-m_+(\tau)}}
    \ge\frac{b\,\hat\beta^\Delta}
             {1+b\,\hat\beta^\Delta},\\
  \Pr(\sigma_v=-1\mid\sigma_{N_G(v)}=\tau)
  &=\frac{1}
          {1+\lambda_v\beta^{m_-(\tau)-m_+(\tau)}}
    \ge\frac{1}{1+r\,\hat\beta^{-\Delta}}.
\end{aligned}
\]
Under any partial pinning that leaves $v$ unpinned, its conditional spin
probabilities are convex combinations of the displayed full-neighborhood
conditional probabilities over the compatible configurations of $N_G(v)$.
Thus both conditional spin probabilities remain at least $\rho$ after every
partial pinning.
Applying \cref{lem:dtv-citation}\textnormal{(ii)} with this value of $\rho$
and using the definition of $s$ gives
\[
  \TV(\mu,\nu)
  \ge\frac{\rho^2}{4}
       \max_{u\in V}
       \frac{\abs{\log\lambda_u^\mu-\log\lambda_u^\nu}}
            {\deg_G(u)+1}
  \ge\frac{\rho^2}{4(\Delta+1)}s.
\]
\end{proof}

\subsection{The oracle reduction}
\label{sec:oracle-reduction}

\begin{proposition}
\label{prop:oracle-to-tv}
Fix real numbers $0<b<r<\infty$ and one of the two model families described
above.  Suppose that there is a constant $\kappa\in(0,1]$ such that, for every
$\lambda^\mu,\lambda^\nu\in[b,r]^V$,
\begin{equation}
  \TV(\pi_{\lambda^\mu},\pi_{\lambda^\nu})
  \ge\kappa
  \norm{\log\lambda^\nu-\log\lambda^\mu}_\infty.
  \label{eq:oracle-reduction-separation}
\end{equation}
Suppose also that, for every $\lambda\in[b,r]^V$ and
$\varepsilon'\in(0,1)$, the following deterministic oracles are available,
running in time $T_{\widehat p}(n,\varepsilon')$ and
$T_{\widehat\phi}(n,\varepsilon')$, respectively, with constants depending
only on the fixed model parameters.
\begin{enumerate}[label=\textnormal{(\roman*)},leftmargin=2.6em]
  \item Given $u\in V$, the marginal-probability oracle returns
  $\widehat p$ satisfying
  \[
    \abs{\widehat p-
      \Pr_{\sigma\sim\pi_\lambda}(\sigma_u=+1)}
    \le\varepsilon'.
  \]

  \item Given coefficients $a_v\in[-1,1]$ and $t\in\R$, define
  $A(\sigma)=\sum_{v:\,\sigma_v=+1}a_v$.
  The characteristic-function oracle returns $\widehat\Phi_{\pi_\lambda,A}(t)$
  satisfying
  \[
    \abs{
      \widehat\Phi_{\pi_\lambda,A}(t)
      -\Phi_{\pi_\lambda,A}(t)
    }
    \le\varepsilon'.
  \]
\end{enumerate}
Then, for every $\lambda^\mu,\lambda^\nu\in[b,r]^V$ and
$\varepsilon\in(0,1)$, the distance
$\TV(\pi_{\lambda^\mu},\pi_{\lambda^\nu})$ can be deterministically
approximated within relative error $\varepsilon$.
The running time is
\[
\begin{aligned}
  \widetilde O\Biggl(
    &\frac{n^3}{\varepsilon\kappa}
      T_{\widehat p}\left(n,\frac{\varepsilon\kappa}{128n}\right)
      +\frac{n^3}{(\varepsilon\kappa)^3}+\frac{n}{\varepsilon\kappa}
      T_{\widehat\phi}\left(
        n,\frac{\varepsilon\kappa}
        {1536n\lceil384n/(\varepsilon\kappa)\rceil+2\varepsilon\kappa}
      \right)
  \Biggr),
\end{aligned}
\]
with the same dependence on the fixed model parameters.
The output is zero if and only if
$\lambda^\mu=\lambda^\nu$.
\end{proposition}

\begin{proof}
The deterministic algorithm is as follows.
\begin{enumerate}[leftmargin=2.4em]
  \item Compare $\lambda^\mu$ and $\lambda^\nu$ exactly.  If they agree,
  return $0$.

  \item Compute $s$, the coefficients $a_v$, and the statistic $A$ as in
  \cref{sec:normalization}.  Let
  \[
    \varepsilon'=\frac{\varepsilon\kappa}{16}.
  \]
  Since $0<\kappa\le1$, we have $\varepsilon'\in(0,1)$.

  \item Use the algorithmic assertion of
  \cref{lem:thermodynamic-identity}, with the marginal-probability oracle in
  assumption~\textnormal{(i)}, to compute $\widehat c\in[-n,n]$ satisfying
  \[
    \abs{\widehat c-c}\le\varepsilon'.
  \]

  \item Define
  \[
    F(x)=g_s(x-\widehat c),
    \qquad x\in[-n,n].
  \]
  By \cref{lem:normalized-tv}, $F$ is $1$-Lipschitz and nonnegative, and
  $F(x)\le(\widehat c-x)_+\le2n$.  It is exactly evaluable in constant time in the
  stipulated arithmetic model.  Apply
  \cref{lem:fourier-expectation}, using the characteristic-function oracle
  in assumption~\textnormal{(ii)}, to compute $\widehat q$ satisfying
  \[
    \abs{
      \widehat q-\E_{\sigma\sim\mu}g_s(A(\sigma)-\widehat c)
    }
    \le\varepsilon'.
  \]

  \item Return
  \[
    \widehat d=s\max\{0,\widehat q\}.
  \]
\end{enumerate}

For the analysis, define
\[
  q=\E_{\sigma\sim\mu}g_s(A(\sigma)-c)
    =\frac{\TV(\mu,\nu)}s.
\]
By \eqref{eq:oracle-reduction-separation}, $q\ge\kappa$.  The $1$-Lipschitz
property of $g_s$ and the guarantee for $\widehat c$ give
\[
  \abs{
    \E_{\sigma\sim\mu}g_s(A(\sigma)-\widehat c)-q
  }
  \le\varepsilon'.
\]
Consequently,
\[
  \abs{\widehat q-q}
  \le2\varepsilon'
  =\frac{\varepsilon\kappa}{8}
  \le\frac{\varepsilon q}{8}.
\]
If the input vectors are distinct, then $\varepsilon<1$ and $q\ge\kappa$
imply $\widehat q>0$.  Therefore
\[
  \abs{\frac{\widehat d}{\TV(\mu,\nu)}-1}
  =\abs{\frac{\widehat q}{q}-1}
  \le\frac\varepsilon8
  \le\varepsilon.
\]

Finally, applying \cref{lem:thermodynamic-identity} with
$\varepsilon'=\varepsilon\kappa/16$ gives $O(n^3/(\varepsilon\kappa))$
marginal-oracle calls, each with accuracy $\varepsilon\kappa/(128n)$.
Since $T_F(n)=O(1)$, \cref{lem:fourier-expectation} uses
$\widetilde O(n^3/(\varepsilon\kappa)^3)$ time for the coefficients and
$O(n/(\varepsilon\kappa))$ characteristic-function oracle calls, each with
accuracy
$\frac{\varepsilon\kappa}
{1536n\lceil384n/(\varepsilon\kappa)\rceil+2\varepsilon\kappa}$.
The remaining operations take $O(n)$ time.  The complete algorithm therefore
runs in time
\[
\begin{aligned}
  \widetilde O\Biggl(
    &\frac{n^3}{\varepsilon\kappa}
      T_{\widehat p}\left(n,\frac{\varepsilon\kappa}{128n}\right)
      +\frac{n^3}{(\varepsilon\kappa)^3}\\
    &+\frac{n}{\varepsilon\kappa}
      T_{\widehat\phi}\left(
        n,\frac{\varepsilon\kappa}
        {1536n\lceil384n/(\varepsilon\kappa)\rceil+2\varepsilon\kappa}
      \right)
  \Biggr),
\end{aligned}
\]
with the asserted parameter dependence.
\end{proof}

\section{The Lee--Yang zero-free regime}
\label{sec:lee-yang-regime}

We apply the framework of \cref{sec:common-framework} to the ferromagnetic
Ising model.  The following lemma extends the deterministic approximation
result of Liu, Sinclair, and Srivastava
\cite[Theorem~1.1]{liuSinclairSrivastava2017ising} from a uniform activity to
nonuniform complex vertex activities.

\begin{lemma}
\label{lem:lee-yang-complex-oracle}
Fix an integer $\Delta\ge1$ and constants $0<\beta\le1$ and $0<r<1$.
There are constants $a_{\mathrm{LY}},b_{\mathrm{LY}}\ge1$, depending only on
$\Delta$ and $r$, and a deterministic algorithm with the following guarantee.
Let $H$ be
any graph of maximum degree at most $\Delta$, let $z\in\C^{V(H)}$ satisfy
\[
  0<\abs{z_v}\le r
  \qquad(v\in V(H)),
\]
and let $\varepsilon\in(0,1)$.  The algorithm returns a complex number
$\widehat Z_H(\beta,z)$ satisfying
\[
  \abs{\widehat Z_H(\beta,z)-Z_H(\beta,z)}
  \le\varepsilon\abs{Z_H(\beta,z)}
\]
in time
\[
  O_{\Delta,\beta,r}\left(
    (1+\abs{V(H)})^{a_{\mathrm{LY}}}
    \varepsilon^{-b_{\mathrm{LY}}}
  \right).
\]
\end{lemma}

The proof restricts the multivariate partition function to the univariate
polynomial
\[
  Q(t)=Z_H\bigl(\beta,(tz_v)_{v\in V(H)}\bigr).
\]
It then combines the multivariate Lee--Yang theorem with the coefficient
algorithm of Liu et al.  For completeness, the details are given in
Appendix~\ref{app:lee-yang-complex-oracle}.

\begin{corollary}
\label{cor:lee-yang-ising-oracles}
Fix the constants in \cref{lem:lee-yang-complex-oracle}, let $H$ have maximum
degree at most $\Delta$, and let $\lambda\in[b,r]^{V(H)}$, where
$0<b<r<1$.  For every $\varepsilon'\in(0,1)$, each of the following
quantities can be deterministically approximated to additive error
$\varepsilon'$ in time
\[
  O_{\Delta,\beta,r}\left(
    (1+\abs{V(H)})^{a_{\mathrm{LY}}}
    (\varepsilon')^{-b_{\mathrm{LY}}}
  \right):
\]
\begin{enumerate}[label=\textnormal{(\roman*)},leftmargin=2.6em]
  \item given $u\in V(H)$, the positive-spin marginal
  \[
    \Pr_{\sigma\sim\mu_{H,\beta,\lambda}}(\sigma_u=+1);
  \]

  \item given coefficients $a_v\in[-1,1]$ and $t\in\R$, the characteristic
  function
  \[
    \Phi_{\mu_{H,\beta,\lambda},A}(t),
    \qquad
    A(\sigma)=\sum_{v:\,\sigma_v=+1}a_v.
  \]
\end{enumerate}
\end{corollary}

\begin{proof}
For the first assertion, condition on $\sigma_u=-1$ and delete $u$.  This
replaces the external field $\lambda_v$ of each neighbor
$v\in N_H(u)$ by $\beta\lambda_v$.  More precisely, for
$v\in V(H)\setminus\{u\}$, set
\[
  \lambda^{u,-}_v=
  \begin{cases}
    \beta\lambda_v,&v\in N_H(u),\\
    \lambda_v,&v\notin N_H(u).
  \end{cases}
\]
Then
\begin{equation}
  \Pr_{\sigma\sim\mu_{H,\beta,\lambda}}(\sigma_u=-1)
  =\frac{Z_{H-u}(\beta,\lambda^{u,-})}
         {Z_H(\beta,\lambda)}.
  \label{eq:lee-yang-negative-marginal}
\end{equation}
Since $0<\beta\le1$, every coordinate of $\lambda^{u,-}$ has modulus at most
$r$.  Apply \cref{lem:lee-yang-complex-oracle} to the two partition functions
in \eqref{eq:lee-yang-negative-marginal}, with relative accuracy
$\varepsilon'/8$.  The exact ratio is a probability, so
\cref{lem:quotient-estimate} gives additive error less than
$\varepsilon'$.  Taking the real part and projecting onto $[0,1]$ cannot
increase this error; if the resulting estimate is $\widehat p_-$, return
$1-\widehat p_-$ for the positive-spin marginal in assertion~\textnormal{(i)}.

For the second assertion, \eqref{eq:characteristic-partition-ratio} becomes
\begin{equation}
  \Phi_{\mu_{H,\beta,\lambda},A}(t)
  =\frac{Z_H\bigl(\beta,
        (\lambda_v\e^{\iu ta_v})_{v\in V(H)}\bigr)}
       {Z_H(\beta,\lambda)}.
  \label{eq:lee-yang-characteristic-ratio}
\end{equation}
Every twisted coordinate has modulus $\lambda_v\le r$.  Relative approximations with accuracy
$\varepsilon'/8$, followed by \cref{lem:quotient-estimate}, therefore give
the claimed additive approximation.

For the running time, let $n=\abs{V(H)}$ and let $T_Z(n,\delta)$ be a
worst-case running-time upper bound for the partition-function algorithm in
\cref{lem:lee-yang-complex-oracle} on graphs with at most $n$ vertices.
Each marginal query uses two partition-function calls at relative accuracy
$\varepsilon'/8$, together with $O_\Delta(n+1)$ operations to construct the
deleted graph and its fields.  Each characteristic-function query likewise
uses two calls at accuracy $\varepsilon'/8$, with $O(n+1)$ operations to
construct the twisted fields.  Thus both oracle costs are bounded by
\[
  2T_Z(n,\varepsilon'/8)+O_\Delta(n+1)
  =O_{\Delta,\beta,r}\left(
    (n+1)^{a_{\mathrm{LY}}}(\varepsilon')^{-b_{\mathrm{LY}}}
  \right),
\]
using the stated upper bound for $T_Z$ and $a_{\mathrm{LY}},b_{\mathrm{LY}}\ge1$.
The bound for the characteristic-function oracle is uniform in $t$ in the
stipulated arithmetic model.
\end{proof}

\begin{proof}[Proof of \cref{thm:ising-lee-yang}]
First suppose that $\lambda^\mu,\lambda^\nu\in[b,r]^V$.
\Cref{lem:ising-parameter-separation} supplies a positive separation constant
$\kappa_{\mathrm I}$, while
\cref{cor:lee-yang-ising-oracles} supplies the marginal-probability and
characteristic-function oracles along the entire geometric interpolation path.
The common reduction gives the claimed multiplicative approximation.
For the running time, write $n=\abs V\ge1$; the empty graph is handled in
constant time.  The marginal and characteristic-function accuracies in
\cref{prop:oracle-to-tv} are respectively
$\Theta(\varepsilon\kappa_{\mathrm I}/n)$ and
$\Theta((\varepsilon\kappa_{\mathrm I})^2/n^2)$.
Since $\kappa_{\mathrm I}$ is a positive constant depending only on
$\Delta,\beta,b,r$, substituting the bounds of
\cref{cor:lee-yang-ising-oracles} into that proposition gives
\begin{equation}
  \widetilde O_{\Delta,\beta,b,r}\left(
    \frac{n^{a_{\mathrm{LY}}+b_{\mathrm{LY}}+3}}
         {\varepsilon^{b_{\mathrm{LY}}+1}}
    +\frac{n^{a_{\mathrm{LY}}+2b_{\mathrm{LY}}+1}}
         {\varepsilon^{2b_{\mathrm{LY}}+1}}
  \right).
  \label{eq:lee-yang-total-time}
\end{equation}
The two terms account for thermodynamic integration and
characteristic-function queries, respectively.  The Jackson coefficient
cost $\widetilde O(n^3/\varepsilon^3)$ is absorbed by the second term,
since $a_{\mathrm{LY}},b_{\mathrm{LY}}\ge1$.  If $b_{\mathrm{LY}}\ge2$,
the second term also dominates the first.

Now suppose that
$\lambda^\mu,\lambda^\nu\in[r^{-1},b^{-1}]^V$.  Global spin reversal
satisfies
\[
  w_{\beta,\lambda}(-\sigma)
  =\left(\prod_{v\in V}\lambda_v\right)
    w_{\beta,\lambda^{-1}}(\sigma).
\]
It therefore maps $\mu_{G,\beta,\lambda}$ to
$\mu_{G,\beta,\lambda^{-1}}$ and preserves total variation distance.  The
reciprocal external fields lie in $[b,r]^V$, so the first case applies.
Computing these reciprocal fields takes $O(n)$ operations, which are
absorbed by \eqref{eq:lee-yang-total-time}.  The zero-output characterization
follows from the common reduction in both
cases.
\end{proof}

\section{The algorithm in the uniqueness regime}
\label{sec:uniqueness-extension}

In this section, we verify the two oracle assumptions of
\cref{prop:oracle-to-tv} for
\[
  r=(1-\eta)\lambda_{\mathrm c}(\Delta).
\]
Vertex-occupation marginals are obtained from a local complex neighborhood of
the positive-real interval.  Additive approximations to characteristic
functions are obtained by combining phase damping with an enumeration of
occupation patterns on the large-phase coordinates.

\subsection{Local complex partition-function oracle}
\label{sec:local-complex}

\begin{lemma}
\label{lem:local-complex-oracle}
Fix $\Delta\ge3$, $\eta\in(0,1)$, and $b>0$, and let
$r=(1-\eta)\lambda_{\mathrm c}(\Delta)$.  Assume $b<r$.  There exist constants
$\theta_0=\theta_0(\Delta,\eta,b)\in(0,\pi]$ and
$\eta_0=\eta_0(\Delta,\eta,b)\in(0,1/2)$ and a deterministic
algorithm with the following guarantee.  Let $H$ be any induced subgraph of a
graph of maximum degree at most $\Delta$, let
$\lambda\in[b,r]^{V(H)}$, and suppose
\begin{equation*}
  z_v=\lambda_v\e^{\iu\theta_v},
  \qquad
  \abs{\theta_v}\le\theta_0
  \quad (v\in V(H)).
\end{equation*}
Then $Z_H(z)\ne0$, and, for every $\varepsilon'\in(0,1)$, the algorithm returns a
complex number $\widehat Z_H(z)$ such that
\begin{equation*}
  \abs{\widehat Z_H(z)-Z_H(z)}
  \le\varepsilon'\abs{Z_H(z)}.
\end{equation*}
For $m=\abs{V(H)}\ge1$, its running time is
\[
  \widetilde O_{\Delta,\eta,b}\left(
    m\left(\frac{m}{\varepsilon'}\right)^{\alpha_{\mathrm{HC}}}
  \right),
\]
where $\alpha_{\mathrm{HC}}=\log(\Delta-1)/[-\log(1-\eta_0)]>1$.
If $m=0$, the algorithm returns $1$ in constant time.
\end{lemma}

We first introduce the self-avoiding-walk tree of Weitz \cite{weitz2006counting}.
\begin{lemma}[\cite{weitz2006counting}]
\label{lem:saw-tree-representation}
Fix a finite graph $G$, an ordering of the neighbors of every vertex, and
$v\in V(G)$.  One can construct a finite tree
$T=T_{\mathrm{SAW}}(G,v)$, rooted at $x_0$, with the standard Weitz boundary
condition and an endpoint map
\[
  \pi:V(T)\longrightarrow V(G),
  \qquad \pi(x_0)=v,
\]
so every tree vertex, and in particular every leaf, corresponds to a vertex
of $G$ and inherits its external field.  Given a feasible pinning
$\tau\in\{-1,+1\}^{\Lambda}$ on
$\Lambda\subseteq V(G)\setminus\{v\}$, stop a branch when it first reaches a
vertex $x$ with $\pi(x)\in\Lambda$ and pin $x$ to $\tau(\pi(x))$.  Denote this
induced leaf pinning, together with the standard Weitz boundary condition, by
$\widehat\tau$.  Then, for every vector $\lambda$ of positive external fields,
under the hard-core measures with the corresponding external fields,
\[
  \Pr\bigl(v\text{ is occupied in }G\mid\tau\bigr)
  =
  \Pr\bigl(x_0\text{ is occupied in }T\mid\widehat\tau\bigr).
\]
Equivalently, the corresponding root occupation ratios $R_{G,v}^{\tau}$ and
$R_{T,x_0}^{\widehat\tau}$ satisfy
$R_{G,v}^{\tau}=R_{T,x_0}^{\widehat\tau}$.  Indeed, the neighbor-by-neighbor
telescoping argument is an identity of rational functions in formal vertex
activities, so it also holds after nonuniform complex substitutions whenever
all relevant denominators are nonzero.  If $G$ has
maximum degree at most $\Delta$, every non-root vertex of $T$ has at most
$\Delta-1$ children.
\end{lemma}

Li et al.~\cite{li2013correlation} established correlation decay for
the hard-core model throughout the uniqueness regime.  Liu et
al.~\cite[Section~2.3]{liuSinclairSrivastava2019correlation} subsequently lifted this
analysis to a small complex neighborhood, obtaining both zero-freeness and
complex correlation-decay estimates.  Building on these results, we derive
below a deterministic approximation algorithm for the partition function,
uniformly for the nonuniform external fields considered here.

\begin{proof}[Proof of \cref{lem:local-complex-oracle}]
For positive real external fields, consider the pruned SAW tree.  We use the
standard convention that an occupied boundary leaf forces its parent to be
unoccupied, after which that parent's descendants are discarded; hence every
remaining pinned vertex has occupation ratio $R_x=0$.  For each remaining
tree vertex $x$, let $T_x$ be the descendant subtree rooted at $x$, with the
inherited external fields and boundary condition, and let $p_x^{T_x}(s)$ be the
marginal probability that the occupation indicator of $x$ equals
$s\in\{0,1\}$ in $T_x$.  Thus $s=1$ corresponds to spin $+1$ (occupied), and
$s=0$ corresponds to spin $-1$ (unoccupied).  Its
occupation ratio $R_x$ is
\[
  R_x:=\frac{p_x^{T_x}(1)}{p_x^{T_x}(0)}.
\]
If an unpinned vertex $x$ has children $x_1,\ldots,x_d$, then
\begin{equation}
  R_x=F_{\lambda_{\pi(x)},d}(R_{x_1},\ldots,R_{x_d}),
  \qquad
  F_{a,d}(R_1,\ldots,R_d)
  :=a\prod_{j=1}^d(1+R_j)^{-1}.
  \label{eq:local-complex-ratio-recurrence}
\end{equation}
Every non-root vertex has $d\leq\Delta-1$, and every real occupation ratio
$R_x$ satisfies $0\leq R_x\leq r$.

We use the potential and transformed recurrence
\[
  \varphi(t)=2\operatorname{arcsinh}\sqrt t,
  \qquad
  \varphi^{-1}(y)=\sinh^2(y/2),
  \qquad
  F_{z,d}^{\varphi}
  =\varphi\circ F_{z,d}\circ
    (\varphi^{-1},\ldots,\varphi^{-1}).
\]
On the analytic branches agreeing with the positive-real ones,
\begin{equation}
  F_{z,d}^{\varphi}(y_1,\ldots,y_d)
  =2\operatorname{arcsinh}\left(
    \sqrt z\prod_{j=1}^d\operatorname{sech}(y_j/2)
  \right).
  \label{eq:local-complex-transformed-recurrence}
\end{equation}
Put $\mathcal I=[0,\varphi(r)]$.  By
\cite[Theorem~2.14]{liuSinclairSrivastava2019correlation}, for each $a\in[b,r]$ and
$d\leq\Delta-1$,
\[
  \norm{\nabla F_{a,d}^{\varphi}(y)}_1<1-\eta_{a,d}
  \qquad(y\in\mathcal I^d),
\]
where the case $d=0$ is immediate.  The function on the left is continuous,
$[b,r]\times\mathcal I^d$ is compact, and there are only finitely many $d$.
Consequently, there is a single
$\eta_0=\eta_0(\Delta,\eta,b)\in(0,1/2)$ for which
\begin{equation}
  \norm{\nabla F_{a,d}^{\varphi}(y)}_1\leq1-2\eta_0
  \quad
  (a\in[b,r],\ y\in\mathcal I^d,\ d\leq\Delta-1).
  \label{eq:local-complex-real-contraction}
\end{equation}
Notice that this is a local estimate: $a$ is the external field of the current
vertex, so the children may have different external fields.

For real $a\in[b,r]$, $y_j\in\mathcal I$, and $d\leq\Delta$, we also have
\[
  \frac{b}{(1+r)^d}
  \leq F_{a,d}(\varphi^{-1}(y_1),\ldots,\varphi^{-1}(y_d))
  \leq r.
\]
For $0\leq d\leq\Delta$, put
\[
  \mathcal K_d:=[b,r]\times\mathcal I^d.
\]
Each $\mathcal K_d$ is compact, and there are only finitely many $d$.  Since
$\varphi^{-1}(\mathcal I)=[0,r]$, for $(a,\mathbf y)\in\mathcal K_d$ we have
\[
  \cosh(y_j/2)>0
  \quad(1\leq j\leq d),
  \qquad
  \sqrt{\frac{b}{(1+r)^d}}
  \leq
  \sqrt a\prod_{j=1}^d\operatorname{sech}(y_j/2)
  \leq\sqrt r.
\]
Thus the square-root and $\operatorname{arcsinh}$ arguments in
\eqref{eq:local-complex-transformed-recurrence} lie in compact subsets of
$(0,\infty)$ on the real sets $\mathcal K_d$, while the denominators
$\cosh(y_j/2)$ are bounded away from zero.  After shrinking the neighborhoods
of these sets, the arguments remain in the open right half-plane, where we take
the principal square root and the principal $\operatorname{arcsinh}$.  Hence
the maps $F_{z,d}^{\varphi}$ are jointly analytic on a common open neighborhood
of the finitely many sets $\mathcal K_d$.

With $\eta_0$ fixed as in \eqref{eq:local-complex-real-contraction}, continuity
of the gradients, the compactness of each $\mathcal K_d$, and the finite number
of possible values of $d$ allow us to choose uniform constants $\rho,\delta>0$ so that the
corresponding closed tubes remain in this common analytic neighborhood and the
first bound below holds.  Let
\[
  \mathcal D_\rho
  :=\{y\in\C:\operatorname{dist}(y,\mathcal I)\leq\rho\}.
\]
After fixing $\rho$ and $\delta$, choose $C=C(\Delta,\eta,b)\geq1$ large enough
to bound uniformly the relevant $z$-derivatives and message derivatives of
$F_{z,d}^{\varphi}$, as well as $(\varphi^{-1})'$.  If $a\in[b,r]$,
$\abs{z-a}\leq\delta$, and $y\in\mathcal D_\rho^d$, the maps in
\eqref{eq:local-complex-transformed-recurrence} are analytic on a neighborhood
of these points and
\begin{align}
  \sum_{j=1}^d
    \abs{\frac{\partial F_{z,d}^{\varphi}}{\partial y_j}(y)}
  &\leq1-\eta_0
  &&(d\leq\Delta-1),                                      \label{eq:local-complex-complex-contraction}\\
  \abs{F_{z,d}^{\varphi}(y)-F_{a,d}^{\varphi}(y)}
  &\leq C\abs{z-a}
  &&(d\leq\Delta),                                        \label{eq:local-complex-external-field-perturbation}\\
  \sum_{j=1}^d
    \abs{\frac{\partial F_{z,d}^{\varphi}}{\partial y_j}(y)}
  &\leq C
  &&(d\leq\Delta).                                        \label{eq:local-complex-root-derivative}
\end{align}
After decreasing $\rho$ if necessary, $\varphi^{-1}$ is $C$-Lipschitz on
$\mathcal D_\rho$ and
\begin{equation}
  1+\varphi^{-1}(y)=\cosh^2(y/2)\neq0
  \qquad(y\in\mathcal D_\rho).
  \label{eq:local-complex-denominator}
\end{equation}
This establishes the uniform compactness-and-analyticity step needed for the
nonuniform external fields and the boundary message $0$.

Fix $u\in V(H)$, construct $T_{\mathrm{SAW}}(H,u)$, and denote its root by
$x_0$, so that $\pi(x_0)=u$.  Starting at its leaves,
let $y_x^{\mathrm R}$ and $y_x$ be the transformed recursion values obtained from
\eqref{eq:local-complex-transformed-recurrence} using, respectively,
$\lambda_{\pi(x)}$ and $z_{\pi(x)}$; a remaining pinned leaf has transformed
recursion values $y_x=y_x^{\mathrm R}=0$.
Set
\[
  \xi:=\max_{v\in V(H)}\abs{z_v-\lambda_v}
  \leq r\theta_0,
  \qquad
  t:=\frac{C\xi}{\eta_0}.
\]
Choose $\theta_0\in(0,\pi]$ so small that
\[
  r\theta_0\leq\delta,
  \qquad
  \frac{Cr\theta_0}{\eta_0}\leq\frac\rho2.
\]
Proceed by induction from the leaves.  Let $x$ be a non-root vertex and assume
that $y_{x_j}\in\mathcal D_\rho$ and
$\abs{y_{x_j}-y_{x_j}^{\mathrm R}}\leq t$ for every child $x_j$.  Set
\[
  a_x:=\lambda_{\pi(x)},
  \qquad z_x:=z_{\pi(x)},
  \qquad
  \mathbf y:=(y_{x_1},\ldots,y_{x_d}),
  \qquad
  \mathbf y^{\mathrm R}:=(y_{x_1}^{\mathrm R},\ldots,y_{x_d}^{\mathrm R}).
\]
Insert $F_{z_x,d}^{\varphi}(\mathbf y^{\mathrm R})$ between
$y_x=F_{z_x,d}^{\varphi}(\mathbf y)$ and
$y_x^{\mathrm R}=F_{a_x,d}^{\varphi}(\mathbf y^{\mathrm R})$.  Since
$\mathcal D_\rho$ is convex, the mean-value formula,
\eqref{eq:local-complex-complex-contraction}, and
\eqref{eq:local-complex-external-field-perturbation} give
\begin{equation}
\begin{aligned}
  \abs{y_x-y_x^{\mathrm R}}
  &\leq
    \abs{F_{z_x,d}^{\varphi}(\mathbf y)
          -F_{z_x,d}^{\varphi}(\mathbf y^{\mathrm R})}
    +\abs{F_{z_x,d}^{\varphi}(\mathbf y^{\mathrm R})
          -F_{a_x,d}^{\varphi}(\mathbf y^{\mathrm R})} \\
  &\leq(1-\eta_0)
      \max_j\abs{y_{x_j}-y_{x_j}^{\mathrm R}}+C\abs{z_x-a_x} \\
  &\leq(1-\eta_0)t+C\xi=t.
\end{aligned}
  \label{eq:local-complex-transformed-value-error}
\end{equation}
Here and below, a maximum over no children is understood to be zero.  The last
equality uses $t=C\xi/\eta_0$.
Moreover, $t\leq\rho/2$ and $y_x^{\mathrm R}\in\mathcal I$, so
$y_x\in\mathcal D_\rho$.  This completes the induction.

The root may have $\Delta$ children.  Using
\eqref{eq:local-complex-external-field-perturbation} and
\eqref{eq:local-complex-root-derivative} at this final update gives
\[
  \abs{y_{x_0}-y_{x_0}^{\mathrm R}}\leq C(t+\xi).
\]
Decrease $\theta_0$, if necessary, so that
\[
  Cr\theta_0\left(\frac{C}{\eta_0}+1\right)\leq\rho.
\]
Then the displayed root error is at most $\rho$; since
$y_{x_0}^{\mathrm R}\in\mathcal I$, it follows that
$y_{x_0}\in\mathcal D_\rho$.  We may now use the Lipschitz property of
$\varphi^{-1}$ to obtain a constant $C_0=C_0(\Delta,\eta,b)$ such that the
recursive root ratios
$\widetilde R_{x_0}(z):=\varphi^{-1}(y_{x_0})$ and
$R_{x_0}(\lambda):=\varphi^{-1}(y_{x_0}^{\mathrm R})$ satisfy
\[
  \abs{\widetilde R_{x_0}(z)-R_{x_0}(\lambda)}\leq C_0\xi.
\]
Decrease $\theta_0$ once more so that $C_0r\theta_0\leq1/2$.  Equations
\eqref{eq:local-complex-denominator} and
\eqref{eq:local-complex-transformed-recurrence} show that all recursive ratios
$\widetilde R_x(z)$ at unpinned vertices satisfy the ordinary ratio recurrence
with nonzero denominators.
Moreover, since $R_{x_0}(\lambda)\geq0$,
\begin{equation}
  \abs{1+\widetilde R_{x_0}(z)}
  \geq1+R_{x_0}(\lambda)
       -\abs{\widetilde R_{x_0}(z)-R_{x_0}(\lambda)}
  \geq\frac12.
  \label{eq:local-complex-root-lower-bound}
\end{equation}

We now use the induction pattern underlying
\cite[Theorem~2.16]{liuSinclairSrivastava2019correlation}.  That theorem is stated for a common
external field; the following argument proves the coordinatewise extension.
The only change is to replace the scalar pair at a tree vertex $x$ by
$(\lambda_{\pi(x)},z_{\pi(x)})$, for which the uniform estimates above give
\eqref{eq:local-complex-transformed-value-error}.

For completeness, induct on the number of vertices of an induced subgraph
$J$.  If $J$ is nonempty and $u\in V(J)$, again denote the root of the SAW tree
for $(J,u)$ by $x_0$.  The induction hypothesis makes
\[
  R_{J,u}(z)
  :=z_u\frac{Z_{J-N_J[u]}(z)}{Z_{J-u}(z)}
\]
well defined.  For every recursive state $(K,w)$ below the root of this SAW
tree, $K$, $K-w$, and $K-N_K[w]$ are proper induced subgraphs of $J$.  Hence
the induction hypothesis gives
$Z_K(z),Z_{K-w}(z),Z_{K-N_K[w]}(z)\neq0$, and
\[
  1+R_{K,w}(z)=\frac{Z_K(z)}{Z_{K-w}(z)}\neq0.
\]
Thus every denominator in the complex SAW recurrence is nonzero, and its
boundary values and updates agree with those defining the recursive ratios
$\widetilde R_x(z)$.  The complex SAW-tree
identity therefore identifies
$R_{J,u}(z)$ with $\widetilde R_{x_0}(z)$.  Hence
\eqref{eq:local-complex-root-lower-bound} and the deletion identity give
\[
  Z_J(z)=Z_{J-u}(z)(1+R_{J,u}(z))\neq0.
\]
The empty graph starts the induction.  Thus the same $\theta_0$ gives
zero-freeness for every admissible induced subgraph.

It remains to describe the approximation algorithm.  To approximate
$R_{H,u}(z)$, truncate the pruned SAW tree at depth $L\geq1$.  Set
$y_x=y_x^{\mathrm R}=0$ at every new leaf and apply the recursions toward the
root.  The real values stay in $\mathcal I$.  The preceding comparison with
the real recursion therefore shows that all complex values stay in
$\mathcal D_\rho$; the root is again handled separately.  Put
\[
  B:=\sup_{y\in\mathcal D_\rho}\abs y<\infty.
\]
At each new leaf created by the truncation, setting its transformed recursion
value to zero introduces an error of at most $B$.  From these leaves up to the
root, the exact and truncated recursions use the same external fields, so
there is no external-field perturbation term: each of the $L-1$
non-root levels contracts the discrepancy by $1-\eta_0$, while the root map and
$\varphi^{-1}$ are
uniformly Lipschitz.  Thus the error is bounded by a constant times
$(1-\eta_0)^{L-1}$.  Absorbing $(1-\eta_0)^{-1}$ into the constant, we obtain an
$A=A(\Delta,\eta,b)\geq1$ such that
\begin{equation}
  \abs{\widehat R_{H,u}(z)-R_{H,u}(z)}
  \leq A(1-\eta_0)^L.
  \label{eq:local-complex-truncation-error}
\end{equation}

Let $m=\abs{V(H)}$.  If $m=0$, return $1$.  Otherwise, order the vertices as
$v_1,\ldots,v_m$ and let $H_j=H-\{v_1,\ldots,v_{j-1}\}$.  Repeated use of the
deletion identity gives
\begin{equation}
  Z_H(z)=\prod_{j=1}^m(1+R_{H_j,v_j}(z)).
  \label{eq:local-complex-self-reduction}
\end{equation}
Choose
\[
  L=\max\left\{1,
    \left\lceil\frac{\log(8Am/\varepsilon')}{-\log(1-\eta_0)}\right\rceil
  \right\}
  =O_{\Delta,\eta,b}\left(\log\frac{m+1}{\varepsilon'}\right)
\]
so that the right-hand side of
\eqref{eq:local-complex-truncation-error} is at most
$\varepsilon'/(8m)$.  Write $R_j=R_{H_j,v_j}(z)$ and let $\widehat R_j$ be its
truncated-tree approximation.  By
\eqref{eq:local-complex-root-lower-bound},
\[
  e_j:=\frac{\widehat R_j-R_j}{1+R_j}
  \quad\text{satisfies}\quad
  \abs{e_j}\leq\frac{\varepsilon'}{4m}.
\]
Therefore the output
\[
  \widehat Z_H(z):=\prod_{j=1}^m(1+\widehat R_j)
\]
satisfies
\[
  \abs{\frac{\widehat Z_H(z)}{Z_H(z)}-1}
  \leq\prod_{j=1}^m(1+\abs{e_j})-1
  \leq\exp(\varepsilon'/4)-1
  \leq\varepsilon'.
\]

A depth-$L$ SAW tree has at most
\[
  1+\Delta\sum_{\ell=0}^{L-1}(\Delta-1)^\ell
  =1+\frac{\Delta((\Delta-1)^L-1)}{\Delta-2}
  =O_\Delta((\Delta-1)^L)
\]
vertices before pruning.  Generate the truncated tree directly by a
depth-first traversal, keeping the current path and fixed neighbor orders.
At each node, scanning the path to detect closed cycles and determine their
Weitz boundary conditions takes $O_\Delta(L)$ time.  An occupied boundary
leaf pins its parent to zero, so that parent's descendants need not be
generated.  Each recursion update takes $O_\Delta(1)$ arithmetic operations.
Preparing each induced subgraph $H_j$ costs $O_\Delta(m)$ time even if it is
built afresh.  Thus the $m$ ratio computations and the final product take
\[
  O_\Delta\left(m^2+mL(\Delta-1)^L\right)
\]
time.  Put
\[
  \alpha_{\mathrm{HC}}
  =\frac{\log(\Delta-1)}{-\log(1-\eta_0)}>1.
\]
By the choice of $L$,
\[
  (\Delta-1)^L
  \le(\Delta-1)
      \left(\frac{8Am}{\varepsilon'}\right)^{\alpha_{\mathrm{HC}}}.
\]
Since $A$ and $\eta_0$ depend only on $\Delta,\eta,b$,
$L=O_{\Delta,\eta,b}(\log((m+1)/\varepsilon'))$, and
$\alpha_{\mathrm{HC}}>1$, the total running time is
\[
  \widetilde O_{\Delta,\eta,b}\left(
    m\left(\frac{m}{\varepsilon'}\right)^{\alpha_{\mathrm{HC}}}
  \right),
\]
as claimed.
\end{proof}

\begin{corollary}
\label{cor:real-marginal-oracle}
Under the hypotheses and notation of \cref{lem:local-complex-oracle}, let
$u\in V(H)$, write $m=\abs{V(H)}\ge1$, and let
$\lambda\in[b,r]^{V(H)}$.  For every
$\varepsilon'\in(0,1)$, the occupation probability
$\Pr_{\sigma\sim\mu_{H,\lambda}}(\sigma_u=+1)$ can be approximated to
additive error $\varepsilon'$ in time
\[
  \widetilde O_{\Delta,\eta,b}\left(
    m^{1+\alpha_{\mathrm{HC}}}
    (\varepsilon')^{-\alpha_{\mathrm{HC}}}
  \right).
\]
\end{corollary}

\begin{proof}
For positive external fields,
\[
  p=\Pr_{\sigma\sim\mu_{H,\lambda}}(\sigma_u=+1)
  =\lambda_u\frac{Z_{H-N_H[u]}(\lambda)}{Z_H(\lambda)}.
\]
Set
\[
  x:=\lambda_u Z_{H-N_H[u]}(\lambda),
  \qquad
  y:=Z_H(\lambda),
\]
so that $p=x/y$.
Apply \cref{lem:local-complex-oracle}, with all phases zero, to the two
partition functions with relative accuracy $\varepsilon'/8$.  Let $\widehat x$
be $\lambda_u$ times the approximation to $Z_{H-N_H[u]}(\lambda)$, and let
$\widehat y$ be the approximation to $Z_H(\lambda)$.  Then
$\abs{\widehat x-x}\le(\varepsilon'/8)x\le(\varepsilon'/8)y$ and
$\abs{\widehat y-y}\le(\varepsilon'/8)y$.  Since $p\le1$,
\cref{lem:quotient-estimate} gives error less than $\varepsilon'$.
Taking the real part and projecting it onto $[0,1]$ cannot increase the
error.  Constructing $H-N_H[u]$ and the restricted field vector takes
$O_\Delta(m)$ time.  There are two partition-function oracle calls, each
with relative accuracy $\varepsilon'/8$ on a graph with at most $m$
vertices; an empty residual graph is handled in constant time.  By
\cref{lem:local-complex-oracle}, these calls take total time
\[
  \widetilde O_{\Delta,\eta,b}\left(
    m\left(\frac{8m}{\varepsilon'}\right)^{\alpha_{\mathrm{HC}}}
  \right).
\]
The remaining arithmetic takes constant time, so the preprocessing and both
oracle calls are covered by the stated bound.
\end{proof}

\subsection{All-phase characteristic-function oracle}
\label{sec:global-phase}

\begin{lemma}
\label{lem:phase-damping}
Fix an integer $\Delta\ge0$ and real numbers $0<b<r<\infty$.  Let
$G=(V,E)$ have maximum degree at most $\Delta$, let
$\lambda\in[b,r]^V$, and let
$\theta=(\theta_v)_{v\in V}\in[-\pi,\pi]^V$.  Fix
$\alpha\in(0,\pi]$ and let
$W_\alpha=\set{v\in V:\abs{\theta_v}\ge\alpha}$.  There is a constant
$\gamma=\gamma(\Delta,b,r,\alpha)>0$ such that
\begin{equation}
  \abs{
    \E_{\sigma\sim\mu_{G,\lambda}}
      \exp\!\left(\iu\sum_{v\in V}
        \theta_v\one_{\{\sigma_v=+1\}}\right)
  }
  \le\exp(-\gamma\abs{W_\alpha}).
  \label{eq:phase-damping}
\end{equation}
\end{lemma}

\begin{proof}
Greedily choose a maximal set $S\subseteq W_\alpha$ in which every pair of
distinct vertices is at graph distance at least four.  A radius-three
ball contains at most
\begin{equation*}
  D_\Delta
  =1+\Delta+\Delta(\Delta-1)+\Delta(\Delta-1)^2
\end{equation*}
vertices, so maximality gives $\abs S\ge\abs{W_\alpha}/D_\Delta$.

For $v\in S$, let $B_v=N_G[v]$.  These blocks are pairwise disjoint, and no
edge joins two different blocks; otherwise their centers would be at distance
at most three.  Let $\sigma\sim\mu_{G,\lambda}$, and let
$\mathcal F_{\mathrm{out}}$ be the sigma-algebra generated by the spins
outside $\bigcup_{v\in S}B_v$.  Conditional on any feasible configuration of
these outside spins, the conditional Gibbs distribution factors across the
blocks.  Inside $B_v$, both the all-negative configuration and the
configuration in which only $v$ is positive remain feasible.  The conditional
block partition function is at most
$(1+r)^{\abs{B_v}}\le(1+r)^{\Delta+1}$, so each of these two
configurations has conditional probability at least
\begin{equation*}
  p_0=\frac{\min\{1,b\}}{(1+r)^{\Delta+1}}.
\end{equation*}

For the conditional spins in $B_v$, define
\[
  U_v=\exp\!\left(
    \iu\sum_{u\in B_v:\,\sigma_u=+1}\theta_u
  \right),
\]
and, conditional on $\mathcal F_{\mathrm{out}}$, let $U_v'$ be an independent
copy of $U_v$.  Let $A_v$ be the event that all spins in $B_v$ are negative,
and let $C_v$ be the event that $v$ is positive and every other spin in $B_v$
is negative.  By the preceding paragraph,
\[
  \Pr(A_v\mid\mathcal F_{\mathrm{out}}),
  \Pr(C_v\mid\mathcal F_{\mathrm{out}})\ge p_0.
\]
Moreover, $U_v=1$ on $A_v$ and $U_v=\e^{\iu\theta_v}$ on $C_v$.
Since $\abs{U_v}=1$ and the two copies are conditionally independent,
\begin{equation*}
\begin{aligned}
  1-\abs{\E[U_v\mid\mathcal F_{\mathrm{out}}]}^2
  &=\frac12\E[\abs{U_v-U_v'}^2
                  \mid\mathcal F_{\mathrm{out}}] \\
  &\ge\frac12\,(2p_0^2)
       \abs{1-\e^{\iu\theta_v}}^2 \\
  &=4p_0^2\sin^2(\theta_v/2)
   \ge4p_0^2\sin^2(\alpha/2).
\end{aligned}
\end{equation*}
Indeed, the inequality keeps only the two ordered cases in which one copy is
in $A_v$ and the other is in $C_v$.  The last inequality follows from
$v\in S\subseteq W_\alpha$ and $\theta_v\in[-\pi,\pi]$.
Since $p_0<1/2$, every conditional characteristic-function factor has modulus
at most
\[
  \rho
  =\sqrt{1-4p_0^2\sin^2(\alpha/2)},
  \qquad 0<\rho<1.
\]
Conditional independence shows that the conditional product of all block
factors has modulus at most $\rho^{\abs S}$.  The phase outside the blocks is
$\mathcal F_{\mathrm{out}}$-measurable and has modulus one, and averaging over
$\mathcal F_{\mathrm{out}}$ cannot increase the bound.  Hence
\eqref{eq:phase-damping} holds with
$\gamma=-\log(\rho)/D_\Delta$.
\end{proof}

\begin{lemma}
\label{lem:additive-phase-oracle}
Fix $\Delta\ge3$, $\eta\in(0,1)$, and $b>0$, and let
$r=(1-\eta)\lambda_{\mathrm c}(\Delta)$.  Assume $b<r$.  Let
$G=(V,E)$ have maximum degree at most $\Delta$, write $n=\abs V$, let
$\lambda\in[b,r]^V$, and let
$\theta=(\theta_v)_{v\in V}$ be a vector of real phases.  For every
$\varepsilon'\in(0,1)$, there is a deterministic algorithm returning a complex number
$\widehat\Phi_\lambda(\theta)$ such that
\begin{equation}
  \abs{
    \widehat\Phi_\lambda(\theta)
    -\E_{\sigma\sim\mu_{G,\lambda}}
       \exp\!\left(\iu\sum_{v\in V}
             \theta_v\one_{\{\sigma_v=+1\}}\right)
  }
  \le\varepsilon'.
  \label{eq:additive-phase-oracle}
\end{equation}
Let $T(n,\delta)\ge n$ bound the running time of the algorithm in
\cref{lem:local-complex-oracle} on graphs with at most $n$ vertices and
relative accuracy $\delta$.
For $n\ge1$, the running time is
\[
  O_{\Delta,\eta,b}\left(
    \left(\frac1{\varepsilon'}\right)^{(\log 2)/\gamma}
    T\left(n,\frac{\varepsilon'}{64}\right)
  \right),
\]
where $\theta_0$ is as in
\cref{lem:local-complex-oracle}, and $\gamma>0$ is the damping constant from
\cref{lem:phase-damping} with $\alpha=\theta_0/2$.
If $n=0$, the algorithm returns $1$ in constant time.
\end{lemma}

\begin{proof}
If $n=0$, return $1$.  Otherwise, assume $n\ge1$.
Write
\[
  \Phi_\lambda(\theta)
  =\E_{\sigma\sim\mu_{G,\lambda}}
      \exp\!\left(\iu\sum_{v\in V}
        \theta_v\one_{\{\sigma_v=+1\}}\right).
\]
Let $\overline\theta=(\overline\theta_v)_{v\in V}\in[-\pi,\pi]^V$ be the
coordinatewise principal representative of $\theta$.  Since the integrand is
$2\pi$-periodic in every coordinate,
$\Phi_\lambda(\theta)=\Phi_\lambda(\overline\theta)$.  Define the large-phase set
\begin{equation*}
  W=\set{v\in V:\abs{\overline\theta_v}\ge\theta_0/2},
\end{equation*}
where $\theta_0$ is the local-tube radius from
\cref{lem:local-complex-oracle}.  Thus the phase at every vertex outside $W$
has modulus strictly less than $\theta_0/2$, whereas the phases at vertices in
$W$ have distance at least $\theta_0/2$ from $2\pi\mathbb Z$.

Apply \cref{lem:phase-damping} with
$\overline\theta$ and $\alpha=\theta_0/2$, and let $\gamma>0$ be its damping
constant.  Define
\begin{equation*}
  K
  =\left\lceil\gamma^{-1}\log(8/\varepsilon')\right\rceil.
\end{equation*}
If $\abs W>K$, the algorithm returns zero.  By
\eqref{eq:phase-damping} and the equality
$\Phi_\lambda(\theta)=\Phi_\lambda(\overline\theta)$,
\[
  \abs{\Phi_\lambda(\theta)}
  \le\e^{-\gamma\abs W}<\varepsilon'/8.
\]

It remains to consider $\abs W\le K$.  Define
$z_v=\lambda_v\e^{\iu\overline\theta_v}$ and partition the feasible
configurations according to
$J=\set{v\in W:\sigma_v=+1}$.  For every $J\in\cI(G[W])$, define
\begin{equation*}
  H_J=G-\bigl(W\cup N_G(J)\bigr).
\end{equation*}
Then the exact decomposition is
\begin{equation}
  Z_G(z)
  =\sum_{J\in\cI(G[W])}
      z^J Z_{H_J}(z|_{V(H_J)}).
  \label{eq:bad-set-decomposition}
\end{equation}
Every vertex of $H_J$ is outside $W$, so its phase is at most
$\theta_0/2$ in modulus.  The safety margin in
\cref{lem:local-complex-oracle} therefore applies to every internal partition
function in \eqref{eq:bad-set-decomposition}.

Relative errors in these terms are not amplified by cancellation.  Indeed,
\begin{equation}
\begin{aligned}
  \sum_{J\in\cI(G[W])}
    \abs{z^J Z_{H_J}(z|_{V(H_J)})}
  &\le\sum_{J\in\cI(G[W])}
       \lambda^J Z_{H_J}(\lambda|_{V(H_J)}) \\
  &=Z_G(\lambda).
\end{aligned}
  \label{eq:no-cancellation-amplification}
\end{equation}
The equality is the positive-real version of the same decomposition according
to the positive spins in $W$.
Approximate every internal complex partition function to relative error
$\varepsilon'/64$.  Summing these errors and using
\eqref{eq:no-cancellation-amplification} bounds the total additive error in the
numerator by $(\varepsilon'/64)Z_G(\lambda)$.  Approximate the positive denominator
$Z_G(\lambda)$ to the same relative accuracy.  More explicitly, define
\[
  Z_\theta=Z_G(z),\qquad Z_0=Z_G(\lambda),
\]
let $\widehat Z_\theta$ be the enumerated sum of the approximate internal
partition functions, and let $\widehat Z_0$ be the denominator approximation.
Since $\Phi_\lambda(\theta)=Z_\theta/Z_0$, the algorithm returns
$\widehat\Phi_\lambda(\theta)=\widehat Z_\theta/\widehat Z_0$.
Then $\abs{\widehat Z_\theta-Z_\theta}\le(\varepsilon'/64)Z_0$,
$\abs{\widehat Z_0-Z_0}\le(\varepsilon'/64)Z_0$, and
$\abs{Z_\theta}\le Z_0$.
\Cref{lem:quotient-estimate} bounds the quotient error by
$2(\varepsilon'/64)/(1-\varepsilon'/64)\le\varepsilon'/16$.

There are at most
\begin{equation}
  2^{\abs W}
  \le2^K
  =2^{\lceil\gamma^{-1}\log(8/\varepsilon')\rceil}
  \le2\left(\frac8{\varepsilon'}\right)^{(\log 2)/\gamma}
  \label{eq:bad-set-enumeration-cost}
\end{equation}
subsets to inspect.  The algorithm enumerates them, discards those that are not
independent, constructs $H_J$, and invokes the local partition-function oracle.
The principal-phase reduction and construction of $W$ take $O_\Delta(n)$ time.
Inspecting each subset, constructing $H_J$, and computing $z^J$ take
$O_\Delta(n)$ time per subset.  There are at most $2^{\abs W}$ internal oracle
calls and one denominator call.  Each call uses relative accuracy
$\varepsilon'/64$ on a graph with at most $n$ vertices, and therefore takes
at most $T(n,\varepsilon'/64)$ time, with empty graphs handled in constant time.

The work outside the oracle calls takes $O_\Delta(n2^{\abs W})$ time.
By \eqref{eq:bad-set-enumeration-cost}, the total running time is
\[
  O_{\Delta,\eta,b}\left(
    \left(\frac1{\varepsilon'}\right)^{(\log 2)/\gamma}
    T\left(n,\frac{\varepsilon'}{64}\right)
  \right).
\]
\end{proof}

\begin{corollary}
\label{cor:uniqueness-characteristic}
Fix $\Delta\ge3$, $\eta\in(0,1)$, and $b>0$, and let
$r=(1-\eta)\lambda_{\mathrm c}(\Delta)$.  Assume $b<r$.  Let
$G=(V,E)$ have maximum degree at most $\Delta$, write $n=\abs V$, let
$\lambda\in[b,r]^V$, let
$\mu=\mu_{G,\lambda}$, and let
$a_v\in[-1,1]$ and
$A(\sigma)=\sum_{v:\,\sigma_v=+1}a_v$.  For every real frequency $t$ and
every $\varepsilon'\in(0,1)$, one can deterministically compute
$\widehat\Phi_{\mu,A}(t)$ satisfying
\begin{equation*}
  \abs{\widehat\Phi_{\mu,A}(t)-\Phi_{\mu,A}(t)}
  \le\varepsilon'
\end{equation*}
in time
\[
  \widetilde O_{\Delta,\eta,b}\left(
    n^{1+\alpha_{\mathrm{HC}}}
    (\varepsilon')^{-(\alpha_{\mathrm{HC}}+\chi_{\mathrm{HC}})}
  \right)
  \qquad(n\ge1),
\]
where
\[
  \alpha_{\mathrm{HC}}:=\frac{\log(\Delta-1)}{-\log(1-\eta_0)}>1,
  \qquad
  \chi_{\mathrm{HC}}:=\frac{\log 2}{\gamma}
\]
and $\eta_0$ is as in \cref{lem:local-complex-oracle}, while
$\gamma$ is the damping constant in
\cref{lem:additive-phase-oracle}.  Both exponents depend only on
$\Delta,\eta,b$.  If $n=0$, the algorithm returns $1$ in constant time.
\end{corollary}

\begin{proof}
Use \cref{lem:additive-phase-oracle} with $\theta_v=ta_v$ and denote its
output by $\widehat\Phi_{\mu,A}(t)$.
\eqref{eq:characteristic-partition-ratio} identifies the requested
characteristic function with the corresponding partition-function ratio.
The all-phase oracle reduces the phases modulo $2\pi$ internally.
For $n\ge1$, constructing the phase vector takes $O(n)$ time, followed by
one call to the all-phase oracle with additive accuracy $\varepsilon'$.
Combining its enumeration bound with
\cref{lem:local-complex-oracle} gives total time
\[
  O(n)+\widetilde O_{\Delta,\eta,b}\left(
    (\varepsilon')^{-\chi_{\mathrm{HC}}}
    n\left(\frac{64n}{\varepsilon'}\right)^{\alpha_{\mathrm{HC}}}
  \right),
\]
which is the claimed bound.  When $n=0$, the statistic is identically zero,
so its characteristic function is $1$.
\end{proof}

\subsection{Proof of
  \texorpdfstring{\cref{thm:uniqueness-main}}{Theorem 1.3}}
\label{sec:main-proof}

\begin{proof}[Proof of \cref{thm:uniqueness-main}]
Let $r=(1-\eta)\lambda_{\mathrm c}(\Delta)$.
\Cref{cor:real-marginal-oracle}, obtained by specializing
\cref{lem:local-complex-oracle} to zero phase, supplies the
marginal-probability oracle required by \cref{prop:oracle-to-tv}.
\Cref{cor:uniqueness-characteristic}, based on the all-phase construction in
\cref{lem:additive-phase-oracle}, supplies the characteristic-function oracle
required there.  Both guarantees are uniform for graphs of
maximum degree at most $\Delta$ and external fields in $[b,r]$.
The separation hypothesis holds with $\kappa=b/5000$ by
\cref{lem:parameter-lower-bound}.

Applying \cref{prop:oracle-to-tv} now gives the stated deterministic
multiplicative approximation.  For its running time, assume $n\ge1$ and put
$q=\varepsilon\kappa$.  The two oracle accuracies in that proposition are
\[
\begin{aligned}
  \delta_p&=\frac{q}{128n}=\Theta(q/n),\\
  \delta_\phi&=
    \frac{q}{1536n\lceil384n/q\rceil+2q}
    =\Theta(q^2/n^2).
\end{aligned}
\]
The explicit bounds in \cref{cor:real-marginal-oracle} and
\cref{cor:uniqueness-characteristic} therefore give
\[
\begin{aligned}
  T_{\widehat p}(n,\delta_p)
  &=\widetilde O_{\Delta,\eta,b}\left(
    n^{1+2\alpha_{\mathrm{HC}}}q^{-\alpha_{\mathrm{HC}}}
  \right),\\
  T_{\widehat\phi}(n,\delta_\phi)
  &=\widetilde O_{\Delta,\eta,b}\left(
    n^{1+3\alpha_{\mathrm{HC}}+2\chi_{\mathrm{HC}}}
    q^{-2\alpha_{\mathrm{HC}}-2\chi_{\mathrm{HC}}}
  \right).
\end{aligned}
\]
The thermodynamic-integration and Fourier-coefficient costs are
$\widetilde O_{\Delta,\eta,b}
(n^{4+2\alpha_{\mathrm{HC}}}/\varepsilon^{1+\alpha_{\mathrm{HC}}})$
and $\widetilde O(n^3/\varepsilon^3)$, respectively.
Both are dominated by the cost of the characteristic-function queries.
Indeed, with the phase threshold $\theta_0/2\le\pi/2$, the choice in the
proof of \cref{lem:phase-damping} has $\rho>1/\sqrt2$ and hence
$\gamma=-\log(\rho)/D_\Delta<\log2$, so $\chi_{\mathrm{HC}}>1$.
Together with $\alpha_{\mathrm{HC}}>1$, $n\ge1$, and $\varepsilon<1$,
this gives the total running time
\begin{equation}
  \widetilde O_{\Delta,\eta,b}\left(
    \frac{n^{2+3\alpha_{\mathrm{HC}}+2\chi_{\mathrm{HC}}}}
         {\varepsilon^{1+2\alpha_{\mathrm{HC}}+2\chi_{\mathrm{HC}}}}
  \right),
  \label{eq:hardcore-total-time}
\end{equation}
where $\kappa=b/5000$ has been absorbed into the dependence on the fixed
parameters.  Since $\alpha_{\mathrm{HC}}$ and $\chi_{\mathrm{HC}}$ also depend
only on these parameters, this is
$\poly_{\Delta,\eta,b}(n,1/\varepsilon)$.
If $n=0$, the two input field vectors are empty and the algorithm returns
$0$ in constant time.
The positivity argument in the proof of \cref{prop:oracle-to-tv}
gives the zero-output characterization.
This completes the proof.
\end{proof}

\section{The antiferromagnetic Ising correlation-decay regime}
\label{sec:ising-correlation-decay-regime}

Throughout this section, fix $\Delta\ge3$ and $\beta>1$ with
$\beta\ne\Delta/(\Delta-2)$, and let
$[\lambda_-,\lambda_+]$ be the fixed interval in
\cref{thm:ising-correlation-decay}.  Thus this interval is contained in one
connected component of $\mathcal U_{\beta,\Delta}$.

For a graph $H=(V(H),E(H))$, a set $W\subseteq V(H)$, and a pinning
$\tau\in\{-1,+1\}^W$, define the pinned partition function, with the external
fields at the pinned vertices omitted, by
\begin{equation}
  Z_H^\tau(\beta,z)
  =
  \sum_{\substack{\sigma\in\{-1,+1\}^{V(H)}\\
                  \sigma|_W=\tau}}
    \beta^{\abs{\set{\{u,v\}\in E(H):\sigma_u\ne\sigma_v}}}
    \prod_{\substack{v\in V(H)\setminus W:\\ \sigma_v=+1}}z_v,
  \qquad
  z\in\C^{V(H)\setminus W}.
  \label{eq:ising-pinned-partition-function}
\end{equation}

\subsection{Local complex partition-function oracle}
\label{sec:ising-cd-local-complex}

\begin{lemma}
\label{lem:ising-cd-local-complex-oracle}
Fix the constants above.  There exist constants
$\theta_0\in(0,\pi]$ and $\eta_0\in(0,1/2)$, depending only on
$\Delta,\beta,\lambda_-,\lambda_+$, and a
deterministic algorithm with the following guarantee.  Let $H$ be any graph
of maximum degree at most $\Delta$, write $n=\abs{V(H)}$, let
$W\subseteq V(H)$, and let
$\tau\in\{-1,+1\}^W$.  Suppose that
\[
  z_v=\lambda_v\e^{\iu\theta_v},
  \qquad
  \lambda_v\in[\lambda_-,\lambda_+],
  \qquad
  \abs{\theta_v}\le\theta_0
  \quad(v\in V(H)\setminus W).
\]
Then $Z_H^\tau(\beta,z)\ne0$, and, for every
$\varepsilon'\in(0,1)$, the algorithm returns a complex number
$\widehat Z_H^\tau(\beta,z)$ such that
\[
  \abs{\widehat Z_H^\tau(\beta,z)-Z_H^\tau(\beta,z)}
  \le\varepsilon'\abs{Z_H^\tau(\beta,z)}.
\]
For $n\ge1$, its running time is
\[
  O_{\Delta,\beta,\lambda_-,\lambda_+}\left(
    n^2\left(\frac n{\varepsilon'}\right)^{\alpha_{\mathrm I}}
  \right).
\]
Here $\alpha_{\mathrm I}:=\log(\Delta-1)/[-\log(1-\eta_0)]>1$.
The guarantee is uniform over $H$, $W$, and $\tau$, with $n$ counting both
pinned and unpinned vertices.  If $n=0$, the algorithm returns $1$ in constant
time.
\end{lemma}

\begin{proof}
Put $\Lambda=[\lambda_-,\lambda_+]$.  By the reduction preceding
\cite[Lemma~2.8]{liuSinclairSrivastava2019correlation}, we may split each pinned
vertex of degree $d$ into $d$ leaves of the same pinned spin, one on each
incident edge.  Since pinned vertices carry no external-field factors, this
preserves $Z_H^\tau(\beta,z)$ and the maximum-degree bound, while increasing
the graph size by at most a factor of $\Delta$.  We apply the same convention
after each subsequent pinning and continue to use the same graph notation.
Thus all pinned vertices may be assumed to be leaves.

For a pinned instance $(G,W,\tau)$ and an unpinned vertex $u$, define,
whenever the denominator is nonzero,
\begin{equation}
  R_{G,u}^\tau(\beta,z)
  :=z_u
    \frac{Z_G^{\tau\cup\{u=+1\}}
      (\beta,z|_{V(G)\setminus(W\cup\{u\})})}
         {Z_G^{\tau\cup\{u=-1\}}
      (\beta,z|_{V(G)\setminus(W\cup\{u\})})}.
  \label{eq:ising-cd-ratio}
\end{equation}
Suppose that $u$ has unpinned neighbors $v_1,\ldots,v_d$.  Let
\[
  d_u^\pm:=\#\{v\in W:v\sim u,\ \tau_v=\pm1\},
  \qquad
  \iota_u:=d_u^- -d_u^+.
\]
Although \cite[Lemma~2.8]{liuSinclairSrivastava2019correlation} is stated with
a common scalar external field, its telescoping proof establishes an identity in the
vertex activities and therefore gives, coordinatewise,
\begin{equation}
  \begin{aligned}
    R_{G,u}^\tau(\beta,z)
    &=\mathcal T_{z_u,d,\iota_u}(r_1,\ldots,r_d),\\
    \mathcal T_{z,d,\iota}(x_1,\ldots,x_d)
    &:=z\beta^\iota\prod_{j=1}^d g_\beta(x_j),\\
    g_\beta(x)&:=\frac{\beta+x}{1+\beta x}.
  \end{aligned}
  \label{eq:ising-cd-ratio-recurrence}
\end{equation}
Here $r_j$ is the ratio at $v_j$ in the split instance $G^{(j)}$, which has one
fewer unpinned vertex and $\deg_{G^{(j)}}(v_j)\le\Delta-1$; each node uses its
own field.

For $x\ge0$ and $\beta>1$,
\[
  \beta^{-1}\le g_\beta(x)\le\beta.
\]
Since $d+\abs{\iota_u}\le\Delta$, every positive-real recursion value with
fields in $\Lambda$ consequently lies in the fixed interval
\begin{equation}
  r_{\min}\le R\le r_{\max},
  \qquad
  r_{\min}:=\lambda_-\beta^{-\Delta}>0,
  \qquad
  r_{\max}:=\lambda_+\beta^\Delta.
  \label{eq:ising-cd-real-ratio-range}
\end{equation}

Using the potential underlying
\cite[Theorem~2.9]{liuSinclairSrivastava2019correlation}, expressed in our
ratio coordinate, set
\[
  \psi(R):=\log\frac{(1+\zeta)R+\zeta}{\zeta R+1+\zeta},
  \qquad R>0,
\]
where $\zeta=\zeta(\beta,\Delta)>0$ is independent of the field, as in
\cite[Section~3]{sinclairSrivastavaThurley2012approximation}.
This field-independence allows the same potential to be used at vertices
with different activities.  This map is
real analytic and strictly increasing.  It is the odds-ratio form of the
probability-coordinate potential of Sinclair et
al.~\cite{sinclairSrivastavaThurley2012approximation}: substituting
$p=R/(1+R)$ into $\log((p+\zeta)/(1-p+\zeta))$ gives $\psi(R)$.  Use the
analytic branches agreeing with the positive-real functions, and set
\[
  \mathcal S_{z,d,\iota}
  :=\psi\circ\mathcal T_{z,d,\iota}\circ
    (\psi^{-1},\ldots,\psi^{-1}),
  \qquad
  \mathcal Y:=[\psi(r_{\min}),\psi(r_{\max})].
\]
All real messages lie in $\mathcal Y$ by
\eqref{eq:ising-cd-real-ratio-range}.

For $M\in\{\Delta-1,\Delta\}$, put
\[
  \mathfrak A_M
  :=\{(d,\iota)\in\mathbb Z_{\ge0}\times\mathbb Z:
       d+\abs\iota\le M\}.
\]
For $(d,\iota)\in\mathfrak A_{\Delta-1}$, the pinned inputs can be
included in the full $(\Delta-1)$-input probability recursion: insert
$|\iota|$ inputs equal to $0$ if $\iota>0$ and to $1$ if $\iota<0$,
and set any remaining dummy inputs to $1/2$.
Indeed, $g_\beta(0)=\beta$, $g_\beta(\infty)=\beta^{-1}$, and
$g_\beta(1)=1$, so this reproduces the factor $\beta^\iota$ without
changing the free-input recursion.  After applying the potential, the
gradient in the free coordinates is a subvector of the full transformed
gradient.  Therefore, by
\cite[Theorem~2.9]{liuSinclairSrivastava2019correlation}, the
probability--ratio change of coordinates above, the compactness of
$\Lambda\times\mathcal Y^d$ for each fixed
$(d,\iota)\in\mathfrak A_{\Delta-1}$, and the finiteness of
$\mathfrak A_{\Delta-1}$, we have
\[
  \Gamma
  :=\max_{\substack{(d,\iota)\in\mathfrak A_{\Delta-1}\\
                     a\in\Lambda,\ \mathbf y\in\mathcal Y^d}}
       \norm{\nabla\mathcal S_{a,d,\iota}(\mathbf y)}_1
  <1,
\]
where the gradient is zero when $d=0$.  Fix
$\eta_0\in(0,1/2)$ such that $\Gamma\le1-2\eta_0$.  With this choice of
$\eta_0$, the exponent in the statement is
\begin{equation}
  \alpha_{\mathrm I}
  :=\frac{\log(\Delta-1)}{-\log(1-\eta_0)}>1.
  \label{eq:ising-cd-runtime-exponent}
\end{equation}
Thus
\begin{equation}
  \norm{\nabla\mathcal S_{a,d,\iota}(\mathbf y)}_1
  \le1-2\eta_0
  \quad
  (a\in\Lambda,\ \mathbf y\in\mathcal Y^d,\ d+\abs\iota\le\Delta-1).
  \label{eq:ising-cd-real-contraction}
\end{equation}

We next give arguments similar to those in
\cite[Corollary~2.10]{liuSinclairSrivastava2019correlation}.  On the common real
message interval $\mathcal Y$, the inverse potential is
\[
  \psi^{-1}(y)
  =\frac{(1+\zeta)\e^y-\zeta}{1+\zeta-\zeta\e^y}.
\]
Hence $\psi^{-1}(\mathcal Y)=[r_{\min},r_{\max}]$.  In particular, for every
$(d,\iota)\in\mathfrak A_\Delta$, $a\in\Lambda$, and
$\mathbf y\in\mathcal Y^d$,
\[
  1+\beta\psi^{-1}(y_j)>0
  \quad(1\le j\le d),
  \qquad
  \mathcal T_{a,d,\iota}
  \bigl(\psi^{-1}(y_1),\ldots,\psi^{-1}(y_d)\bigr)
  \in[r_{\min},r_{\max}].
\]
Consequently, on each compact set
\[
  \mathcal K_{d,\iota}:=\Lambda\times\mathcal Y^d,
  \qquad (d,\iota)\in\mathfrak A_\Delta,
\]
all denominators in $\psi^{-1}$ and $\mathcal S_{z,d,\iota}$ are bounded away
from zero, while all logarithm arguments lie in a compact subset of
$(0,\infty)$.  After shrinking the neighborhoods, these arguments remain in
the open right half-plane, where we take the principal logarithm.  Thus the
resulting maps are jointly analytic on a common open neighborhood of the sets
$\mathcal K_{d,\iota}$.  With $\eta_0$ fixed as above, continuity of the
gradients, the compactness of each $\mathcal K_{d,\iota}$, and the finiteness
of $\mathfrak A_\Delta$ allow us to choose uniform constants
$\rho,\delta>0$ so that the corresponding closed tubes remain in this common
neighborhood and the first bound below holds.
After fixing such $\rho$ and $\delta$, choose $C\ge1$ large enough to bound the
relevant $z$-derivatives and message derivatives, as well as
$(\psi^{-1})'$.  Let
\[
  \mathcal Y_\rho
  :=\{y\in\C:\operatorname{dist}(y,\mathcal Y)\le\rho\}.
\]
After decreasing $\rho$ and $\delta$ if necessary, whenever $a\in\Lambda$,
$\abs{z-a}\le\delta$, and $\mathbf y\in\mathcal Y_\rho^d$, we have
\begin{align}
  \norm{\nabla\mathcal S_{z,d,\iota}(\mathbf y)}_1
  &\le1-\eta_0
  &&(d+\abs\iota\le\Delta-1),
  \label{eq:ising-cd-complex-contraction}\\
  \abs{\mathcal S_{z,d,\iota}(\mathbf y)
        -\mathcal S_{a,d,\iota}(\mathbf y)}
  &\le C\abs{z-a}
  &&(d+\abs\iota\le\Delta),
  \label{eq:ising-cd-field-perturbation}\\
  \norm{\nabla\mathcal S_{z,d,\iota}(\mathbf y)}_1
  &\le C
  &&(d+\abs\iota\le\Delta).
  \label{eq:ising-cd-root-derivative}
\end{align}
Moreover, $\psi^{-1}$ is $C$-Lipschitz on $\mathcal Y_\rho$.  Shrinking $\rho$
if necessary,
\begin{equation}
  1+\beta\psi^{-1}(y)\ne0
  \quad(y\in\mathcal Y_\rho).
  \label{eq:ising-cd-denominators}
\end{equation}

Iterating the recurrence \eqref{eq:ising-cd-ratio-recurrence} from \cite[Lemma~2.8]{liuSinclairSrivastava2019correlation}
gives a finite computation tree because each step removes an unpinned vertex.
Here
$d+\abs\iota\le\Delta-1$ at non-root nodes, while the
root may have $d+\abs\iota=\Delta$.  Denote the root by $x_0$.  For each node
$x$, let $\pi(x)$ be the graph vertex it represents; in particular,
$\pi(x_0)=u$.  Let $y_x^{\mathrm R}$ and $y_x$ be its transformed messages
under fields
$a_x=\lambda_{\pi(x)}$ and $z_x=z_{\pi(x)}$, respectively.  For real fields,
all restricted partition functions are positive, so
$\psi^{-1}(y_x^{\mathrm R})$ is the corresponding graph ratio.
Put
\[
  \xi:=\max_{v\in V(G)\setminus W}\abs{z_v-\lambda_v}
  \le\lambda_+\theta_0,
  \qquad
  t:=\frac{C\xi}{\eta_0}.
\]
Choose $\theta_0\in(0,\pi]$ small enough that
\begin{equation}
  \lambda_+\theta_0\le\delta,
  \qquad
  \frac{C\lambda_+\theta_0}{\eta_0}\le\frac\rho2.
  \label{eq:ising-cd-theta-first-choice}
\end{equation}
Starting at the leaves, the mean-value formula on the convex set
$\mathcal Y_\rho$, together with \eqref{eq:ising-cd-complex-contraction} and
\eqref{eq:ising-cd-field-perturbation}, gives at every non-root node
\begin{align}
  \abs{y_x-y_x^{\mathrm R}}
  &\le(1-\eta_0)
       \max_j\abs{y_{x_j}-y_{x_j}^{\mathrm R}}+C\xi \notag\\
  &\le(1-\eta_0)t+C\xi=t.
  \label{eq:ising-cd-nonroot-stability}
\end{align}
Here the maximum is $0$ at a leaf.  Since $y_x^{\mathrm R}\in\mathcal Y$ and
$t\le\rho/2$, this also keeps every non-root complex message in
$\mathcal Y_\rho$.

At the root, \eqref{eq:ising-cd-field-perturbation} and
\eqref{eq:ising-cd-root-derivative} instead yield
\begin{equation}
  \abs{y_{x_0}-y_{x_0}^{\mathrm R}}
  \le C(t+\xi).
  \label{eq:ising-cd-root-message-stability}
\end{equation}
After shrinking $\theta_0$, this bound is at most $\rho/2$, so the root message
also lies in $\mathcal Y_\rho$.  Since $\psi^{-1}$ is Lipschitz there, the formal
root ratio
$\widetilde R_{G,u}(\beta,z):=\psi^{-1}(y_{x_0})$ satisfies, for some
parameter-dependent constant $C_0$,
\begin{equation}
  \abs{\widetilde R_{G,u}(\beta,z)
        -R_{G,u}^\tau(\beta,\lambda)}
  \le C_0\xi.
  \label{eq:ising-cd-root-ratio-stability}
\end{equation}
The real ratio is well defined.  Decreasing $\theta_0$ once more so that
$C_0\lambda_+\theta_0\le r_{\min}/2$, and using
\eqref{eq:ising-cd-real-ratio-range},
\begin{equation}
  \Re\widetilde R_{G,u}(\beta,z)\ge\frac{r_{\min}}2>0,
  \qquad
  \abs{1+\widetilde R_{G,u}(\beta,z)}\ge1.
  \label{eq:ising-cd-ratio-lower-bound}
\end{equation}

We now use simultaneous induction on the number $m$ of unpinned vertices to
show that every pinned partition function is nonzero and every ratio
\eqref{eq:ising-cd-ratio} equals its formal computation-tree value (hence
satisfies \eqref{eq:ising-cd-ratio-lower-bound}).  For $m=0$, the partition
function is the nonzero monomial $\beta^q$.

Assume both claims through $m-1$ and take an instance with $m\ge1$ and root
$u$.  Pinning $u$ leaves $m-1$ unpinned vertices, so the two restricted
partition functions in \eqref{eq:ising-cd-ratio} are nonzero.  Each child
instance in \eqref{eq:ising-cd-ratio-recurrence} also has $m-1$ unpinned
vertices.  Substituting the inductively identified child ratios into
\eqref{eq:ising-cd-ratio-recurrence}, which is well defined by
\eqref{eq:ising-cd-denominators}, identifies the parent ratio with its formal
root value.  Finally,
\begin{equation}
  Z_G^\tau(\beta,z)
  =Z_G^{\tau\cup\{u=-1\}}(\beta,z)
     \bigl(1+R_{G,u}^\tau(\beta,z)\bigr),
  \label{eq:ising-cd-deletion-identity}
\end{equation}
where the fields are restricted to the still-unpinned vertices.  Its first
factor is nonzero by induction and its second by
\eqref{eq:ising-cd-ratio-lower-bound}, completing the induction and proving
zero-freeness for every $H,W,\tau$.

To approximate a ratio at $u$, truncate its computation tree at depth $L$ and
assign every new leaf the fixed transformed message
\[
  y_\star:=\psi(r_{\min})\in\mathcal Y.
\]
The real truncated messages stay in $\mathcal Y$ by
\eqref{eq:ising-cd-real-ratio-range}; the stability argument above, including
the root estimate, then keeps the complex truncated messages in
$\mathcal Y_\rho$.  Comparing the exact and truncated recursions, we find that
each non-root level contracts by $1-\eta_0$, while the root map and $\psi^{-1}$ have bounded
derivatives.  The leaf discrepancy is uniformly bounded because
$\mathcal Y_\rho$ is bounded.  Hence, for some
$A=A(\Delta,\beta,\lambda_-,\lambda_+)\ge1$,
\begin{equation}
  \abs{\widehat R_{G,u}(\beta,z)-R_{G,u}^\tau(\beta,z)}
  \le A(1-\eta_0)^L.
  \label{eq:ising-cd-truncation-error}
\end{equation}

Finally, enumerate the $m=\abs{V(H)\setminus W}$ unpinned vertices as
$u_1,\ldots,u_m$.  Put $\tau_0=\tau$, and for $1\le j\le m$ let
\[
  \tau_j=\tau_{j-1}\cup\{u_j=-1\},
  \qquad
  R_j:=R_{H,u_j}^{\tau_{j-1}}(\beta,z),
\]
with the fields restricted to the still-unpinned vertices as in
\eqref{eq:ising-cd-ratio}.
Repeated use of \eqref{eq:ising-cd-deletion-identity} gives
\begin{equation}
  Z_H^\tau(\beta,z)
  =Z_H^{\tau_m}(\beta,z)\prod_{j=1}^m(1+R_j)
  =\beta^{q(\tau_m)}\prod_{j=1}^m(1+R_j),
  \label{eq:ising-cd-self-reduction}
\end{equation}
where $q(\tau_m)$ is the number of disagreeing edges in the unique fully
pinned configuration.

If $m=0$, the algorithm counts the disagreeing edges in $O_\Delta(n)$ time
and returns this exact value; the empty graph is handled in constant time.
Otherwise, choose
\[
  L=\max\left\{1,
    \left\lceil
      \frac{\log(8Am/\varepsilon')}{-\log(1-\eta_0)}
    \right\rceil\right\}.
\]
Approximate every $R_j$ by the depth-$L$ recursion and return
\[
  \widehat Z_H^\tau(\beta,z)
  :=\beta^{q(\tau_m)}\prod_{j=1}^m(1+\widehat R_j).
\]
By \eqref{eq:ising-cd-ratio-lower-bound}, $\abs{1+R_j}\ge1$, while
\eqref{eq:ising-cd-truncation-error} gives
\[
  \abs{e_j}\le\frac{\varepsilon'}{8m},
  \qquad
  e_j:=\frac{\widehat R_j-R_j}{1+R_j}.
\]
It follows from \eqref{eq:ising-cd-self-reduction} that
\[
  \abs{\frac{\widehat Z_H^\tau(\beta,z)}
                 {Z_H^\tau(\beta,z)}-1}
  \le\prod_{j=1}^m(1+\abs{e_j})-1
  \le\exp(\varepsilon'/8)-1
  \le\varepsilon'.
\]

The root of a computation tree has at most $\Delta$ children, and every other
node has at most $\Delta-1$.  A depth-$L$ tree therefore has at most
\[
  N_L:=1+\Delta\sum_{j=0}^{L-1}(\Delta-1)^j
  =O_\Delta\bigl((\Delta-1)^L\bigr)
  =O_{\Delta,\beta,\lambda_-,\lambda_+}\left(
    \left(\frac m{\varepsilon'}\right)^{\alpha_{\mathrm I}}
  \right)
\]
nodes.  Throughout the successive pinnings and splittings, each original
vertex either remains unpinned or contributes at most $\Delta$ pinned leaves.
Thus every recursive instance has $O_\Delta(n)$ vertices and edges.  Preparing
the initial instance takes $O_\Delta(n)$ time; even if we copy the entire
instance at every tree node, constructing its child instances and evaluating
its recurrence take $O_\Delta(n)$ time per node.  We construct at most $m$
trees, and counting $q(\tau_m)$ and forming the final product take
$O_\Delta(n+m)$ additional time.  Hence the total running time is
\[
\begin{aligned}
  &O_{\Delta,\beta,\lambda_-,\lambda_+}\left(
    n+mn\left(\frac m{\varepsilon'}\right)^{\alpha_{\mathrm I}}
  \right)\\
  &\qquad=O_{\Delta,\beta,\lambda_-,\lambda_+}\left(
    n^2\left(\frac n{\varepsilon'}\right)^{\alpha_{\mathrm I}}
  \right),
\end{aligned}
\]
using $1\le m\le n$.  The case $m=0$ also satisfies the claimed bound for
$n\ge1$.
\end{proof}

\begin{corollary}
\label{cor:ising-cd-marginal}
Fix the constants above.  Let $H$ have maximum degree at most $\Delta$, and
let $\lambda\in[\lambda_-,\lambda_+]^{V(H)}$.  Write $n=\abs{V(H)}$.
For every $u\in V(H)$ and
$\varepsilon'\in(0,1)$, the positive-spin marginal
\[
  \Pr_{\sigma\sim\mu_{H,\beta,\lambda}}(\sigma_u=+1)
\]
can be deterministically approximated to additive error $\varepsilon'$ in
time
\[
  O_{\Delta,\beta,\lambda_-,\lambda_+}\left(
    n^2\left(\frac n{\varepsilon'}\right)^{\alpha_{\mathrm I}}
  \right).
\]
\end{corollary}

\begin{proof}
Let $\tau^+$ be the pinning of $u$ to $+1$.  By
\eqref{eq:ising-pinned-partition-function},
\[
  \Pr_{\sigma\sim\mu_{H,\beta,\lambda}}(\sigma_u=+1)
  =
  \lambda_u
  \frac{Z_H^{\tau^+}
    (\beta,\lambda|_{V(H)\setminus\{u\}})}
       {Z_H(\beta,\lambda)}.
\]
Apply \cref{lem:ising-cd-local-complex-oracle}, with all phases zero, to the
pinned numerator and the unpinned denominator, using relative accuracy
$\varepsilon'/8$.  Multiplication by the exact factor $\lambda_u$ preserves
the relative error.  Since the exact ratio is a probability,
\cref{lem:quotient-estimate} gives additive error less than
$\varepsilon'$.  Taking the real part and projecting onto $[0,1]$ cannot
increase the error.  Preparing the two instances takes $O_\Delta(n)$ time.
The two partition-function calls each use relative accuracy
$\varepsilon'/8$, and the remaining arithmetic takes constant time, giving
the displayed bound by \cref{lem:ising-cd-local-complex-oracle}.
\end{proof}

\subsection{All-phase characteristic-function oracle}
\label{sec:ising-cd-global-phase}

\begin{lemma}
\label{lem:ising-cd-phase-damping}
Fix the constants above, let $H=(V,E)$ have maximum degree at most $\Delta$,
and let $\lambda\in[\lambda_-,\lambda_+]^V$.  For
$\theta\in[-\pi,\pi]^V$ and $\alpha\in(0,\pi]$, let
\[
  W_\alpha=\set{v\in V:\abs{\theta_v}\ge\alpha}.
\]
There is a constant $\gamma>0$, depending only on
$\Delta,\beta,\lambda_-,\lambda_+$, and $\alpha$, such that
\begin{equation}
  \abs{
    \E_{\sigma\sim\mu_{H,\beta,\lambda}}
      \exp\!\left(\iu\sum_{v\in V}
        \theta_v\one_{\{\sigma_v=+1\}}\right)
  }
  \le\exp(-\gamma\abs{W_\alpha}).
  \label{eq:ising-cd-phase-damping}
\end{equation}
\end{lemma}

\begin{proof}
Set
\[
  \rho
  =
  \min\left\{
    \frac{\lambda_-\beta^{-\Delta}}
         {1+\lambda_-\beta^{-\Delta}},
    \frac{1}{1+\lambda_+\beta^\Delta}
  \right\}>0.
\]
Choose a maximal independent set $S$ of $H[W_\alpha]$.  Since every vertex
of $W_\alpha$ lies in $S$ or has a neighbor in $S$,
\[
  \abs S\ge\frac{\abs{W_\alpha}}{\Delta+1}.
\]
Condition on all spins outside $S$.  The spins at vertices of $S$ are then
independent.  For $v\in S$, write
\[
  q_v
  =\Pr(\sigma_v=+1\mid\sigma_{V\setminus S}).
\]
The conditional-probability bounds established in the proof of
\cref{lem:ising-parameter-separation}, specialized to $\beta>1$, give
$q_v,1-q_v\ge\rho$.  Therefore
\begin{align*}
  \abs{(1-q_v)+q_v\e^{\iu\theta_v}}^2
  &=1-4q_v(1-q_v)\sin^2(\theta_v/2)\\
  &\le1-4\rho^2\sin^2(\alpha/2).
\end{align*}
By conditional independence, the modulus of the conditional characteristic
function is bounded by
$\bigl(1-4\rho^2\sin^2(\alpha/2)\bigr)^{\abs S/2}$.
Averaging over the spins outside $S$ cannot increase this bound.  Hence
\eqref{eq:ising-cd-phase-damping} holds with
\[
  \gamma
  =-\frac{1}{2(\Delta+1)}
    \log\!\left(1-4\rho^2\sin^2(\alpha/2)\right)>0.
\]
\end{proof}

\begin{lemma}
\label{lem:ising-cd-additive-phase-oracle}
Let $H=(V,E)$ have maximum degree at most
$\Delta$, write $n=\abs V$, let $\lambda\in[\lambda_-,\lambda_+]^V$, and let
$\theta=(\theta_v)_{v\in V}$ be a vector of real phases.  For every
$\varepsilon'\in(0,1)$, there is a deterministic algorithm returning a
complex number $\widehat\Phi_\lambda(\theta)$ such that
\begin{equation}
  \abs{
    \widehat\Phi_\lambda(\theta)
    -
    \E_{\sigma\sim\mu_{H,\beta,\lambda}}
      \exp\!\left(\iu\sum_{v\in V}
        \theta_v\one_{\{\sigma_v=+1\}}\right)
  }
  \le\varepsilon'.
  \label{eq:ising-cd-additive-phase-oracle}
\end{equation}
Let $T(n,\delta)\ge n$ bound the running time of the algorithm in
\cref{lem:ising-cd-local-complex-oracle} on graphs with at most $n$ vertices
and relative accuracy $\delta$, uniformly over all pinnings.
For $n\ge1$, its running time is
\[
  O_{\Delta,\beta,\lambda_-,\lambda_+}\left(
    \left(\frac1{\varepsilon'}\right)^{(\log 2)/\gamma}
    T\left(n,\frac{\varepsilon'}{64}\right)
  \right),
\]
where $\theta_0$ is as in \cref{lem:ising-cd-local-complex-oracle}, and
$\gamma>0$ is the damping constant from \cref{lem:ising-cd-phase-damping}
with $\alpha=\theta_0/2$.
If $n=0$, the algorithm returns $1$ in constant time.
\end{lemma}

\begin{proof}
If $n=0$, return $1$.  Otherwise, assume $n\ge1$.
Let $\overline\theta\in[-\pi,\pi]^V$ be the coordinatewise principal
representative of $\theta$.  The characteristic function is unchanged by
this reduction.  Define
\[
  W=\set{v\in V:\abs{\overline\theta_v}\ge\theta_0/2},
\]
where $\theta_0$ is supplied by
\cref{lem:ising-cd-local-complex-oracle}.  Apply
\cref{lem:ising-cd-phase-damping} with $\alpha=\theta_0/2$, and let
$\gamma>0$ be the resulting constant.  Set
\[
  L
  =\left\lceil
     \gamma^{-1}\log(8/\varepsilon')
    \right\rceil.
\]
If $\abs W>L$, the algorithm returns zero.  In this case
\eqref{eq:ising-cd-phase-damping} bounds the additive error by
$\varepsilon'/8$.

Suppose now that $\abs W\le L$, and set
\[
  z_v=\lambda_v\e^{\iu\overline\theta_v}
  \qquad(v\in V).
\]
For every $\tau\in\{-1,+1\}^W$, write
\[
  c_\tau=\prod_{v\in W:\tau_v=+1}z_v.
\]
Partitioning the configurations according to their spins on $W$ gives
\begin{equation}
  Z_H(\beta,z)
  =
  \sum_{\tau\in\{-1,+1\}^W}
    c_\tau
    Z_H^\tau
      \bigl(\beta,z|_{V\setminus W}\bigr).
  \label{eq:ising-cd-bad-phase-decomposition}
\end{equation}
Every field supplied to the pinned partition functions on the right-hand
side has phase magnitude less than $\theta_0/2$.  Hence
\cref{lem:ising-cd-local-complex-oracle} applies to every term.

Moreover, the triangle inequality and positivity of all coefficients give
\begin{align}
  \sum_{\tau\in\{-1,+1\}^W}
    \abs{
      c_\tau
      Z_H^\tau
        \bigl(\beta,z|_{V\setminus W}\bigr)
    }
  &\le
  \sum_{\tau\in\{-1,+1\}^W}
    \left(\prod_{v\in W:\tau_v=+1}\lambda_v\right)
    Z_H^\tau
      \bigl(\beta,\lambda|_{V\setminus W}\bigr)\notag\\
  &=Z_H(\beta,\lambda).
  \label{eq:ising-cd-no-cancellation}
\end{align}
Approximate every pinned partition function in
\eqref{eq:ising-cd-bad-phase-decomposition} to relative error
$\varepsilon'/64$, and let $\widehat Z_\theta$ be the resulting sum.
Also approximate $Z_H(\beta,\lambda)$ to the same relative accuracy using the
empty pinning, and denote the result by $\widehat Z_0$.  Then
\[
  \abs{\widehat Z_\theta-Z_H(\beta,z)}
  \le\frac{\varepsilon'}{64}Z_H(\beta,\lambda),
  \qquad
  \abs{\widehat Z_0-Z_H(\beta,\lambda)}
  \le\frac{\varepsilon'}{64}Z_H(\beta,\lambda).
\]
The algorithm returns
$\widehat\Phi_\lambda(\theta)=\widehat Z_\theta/\widehat Z_0$.
By \eqref{eq:ising-cd-no-cancellation},
$\abs{Z_H(\beta,z)}\le Z_H(\beta,\lambda)$, so
\cref{lem:quotient-estimate} gives the claimed additive error.

Finally,
\[
  2^{\abs W}
  \le2^L
  =2^{\lceil\gamma^{-1}\log(8/\varepsilon')\rceil}
  \le
  2\left(\frac8{\varepsilon'}\right)^{
       (\log 2)/\gamma}.
\]
The construction of $W$ takes $O_\Delta(n)$ time.
For each assignment $\tau$, preparing the pinned instance and computing
$c_\tau$ take $O_\Delta(n)$ time.  There are $2^{\abs W}$ pinned oracle calls
and one denominator call with empty pinning.  Each call uses relative accuracy
$\varepsilon'/64$ on a graph with $n$ vertices and therefore takes at most
$T(n,\varepsilon'/64)$ time.  Including the enumeration, weighted sum, and
final division, the displayed enumeration bound and
$T(n,\varepsilon'/64)\ge n$ give total running time
\[
  O_{\Delta,\beta,\lambda_-,\lambda_+}\left(
    \left(\frac1{\varepsilon'}\right)^{(\log 2)/\gamma}
    T\left(n,\frac{\varepsilon'}{64}\right)
  \right).
\]
\end{proof}

\begin{corollary}
\label{cor:ising-cd-characteristic}
Fix the constants above.  Let $H=(V,E)$ have maximum degree at most
$\Delta$, write $n=\abs V$, let $\lambda\in[\lambda_-,\lambda_+]^V$, and let
\[
  A(\sigma)=\sum_{v:\,\sigma_v=+1}a_v,
  \qquad
  a_v\in[-1,1].
\]
For every $t\in\R$ and $\varepsilon'\in(0,1)$,
$\Phi_{\mu_{H,\beta,\lambda},A}(t)$ can be deterministically approximated to
additive error $\varepsilon'$.  For $n\ge1$, its running time is
\[
  O_{\Delta,\beta,\lambda_-,\lambda_+}\left(
    \frac{n^{2+\alpha_{\mathrm I}}}
         {(\varepsilon')^{\alpha_{\mathrm I}+\chi_{\mathrm I}}}
  \right),
\]
where
\[
  \alpha_{\mathrm I}:=\frac{\log(\Delta-1)}{-\log(1-\eta_0)}>1,
  \qquad
  \chi_{\mathrm I}:=\frac{\log 2}{\gamma}>0,
\]
and $\eta_0$ is as in \cref{lem:ising-cd-local-complex-oracle}, while
$\gamma$ is the damping constant in
\cref{lem:ising-cd-additive-phase-oracle}.  If $n=0$, the
algorithm returns $1$ in constant time.
\end{corollary}

\begin{proof}
Apply \cref{lem:ising-cd-additive-phase-oracle} with
$\theta_v=ta_v$.  The identity
\eqref{eq:characteristic-partition-ratio} identifies its target with
$\Phi_{\mu_{H,\beta,\lambda},A}(t)$.  Constructing the phases takes $O(n)$
time.  Substituting the bound from \cref{lem:ising-cd-local-complex-oracle}
for $T(n,\varepsilon'/64)$ in
\cref{lem:ising-cd-additive-phase-oracle} gives the displayed running time.
\end{proof}

\subsection{Proof of
  \texorpdfstring{\cref{thm:ising-correlation-decay}}{Theorem 1.5}}
\label{sec:ising-cd-proof}

\begin{proof}[Proof of \cref{thm:ising-correlation-decay}]
If $\lambda_-=\lambda_+$, then the two input vectors agree and the algorithm
returns zero.  Hence assume $\lambda_-<\lambda_+$ and set
$b=\lambda_-$ and $r=\lambda_+$.  For the coordinatewise geometric
interpolation used in \cref{sec:normalization},
\[
  \log\lambda_v(t)
  =(1-t)\log\lambda_v^\mu+t\log\lambda_v^\nu
  \in[\log\lambda_-,\log\lambda_+]
\]
for every $v\in V$ and $t\in[0,1]$.  Thus every coordinate remains in
$[\lambda_-,\lambda_+]$ throughout the path, and this interval is contained
in one connected component of $\mathcal U_{\beta,\Delta}$.

\Cref{cor:ising-cd-marginal} supplies the marginal-probability oracle required
by \cref{prop:oracle-to-tv}, and \cref{cor:ising-cd-characteristic} supplies the
characteristic-function oracle required there.  The separation hypothesis
follows from
\cref{lem:ising-parameter-separation}; its constant
$\kappa_{\mathrm I}>0$ depends only on
$\Delta,\beta,\lambda_-,\lambda_+$.  Set
$\kappa=\min\{1,\kappa_{\mathrm I}\}$.  For $n\ge1$,
\cref{prop:oracle-to-tv} uses marginal accuracy
$\delta_p=\Theta(\varepsilon\kappa/n)$ and characteristic-function accuracy
$\delta_\phi=\Theta((\varepsilon\kappa)^2/n^2)$.
The two oracle bounds above give
\[
\begin{aligned}
  T_{\widehat p}(n,\delta_p)
  &=O_{\Delta,\beta,\lambda_-,\lambda_+}\left(
    \frac{n^{2+2\alpha_{\mathrm I}}}
         {(\varepsilon\kappa)^{\alpha_{\mathrm I}}}
  \right),\\
  T_{\widehat\phi}(n,\delta_\phi)
  &=O_{\Delta,\beta,\lambda_-,\lambda_+}\left(
    \frac{n^{2+3\alpha_{\mathrm I}+2\chi_{\mathrm I}}}
         {(\varepsilon\kappa)^{2\alpha_{\mathrm I}+2\chi_{\mathrm I}}}
  \right).
\end{aligned}
\]
After absorbing the fixed constant $\kappa$, the thermodynamic-integration
and Fourier-coefficient costs are
$\widetilde O_{\Delta,\beta,\lambda_-,\lambda_+}
(n^{5+2\alpha_{\mathrm I}}/\varepsilon^{1+\alpha_{\mathrm I}})$
and $\widetilde O(n^3/\varepsilon^3)$, respectively.
The damping constant in \cref{lem:ising-cd-phase-damping}, evaluated at
$\alpha=\theta_0/2$, satisfies
$\gamma\le\log2/[2(\Delta+1)]$, since $\rho\le1/2$ and $\theta_0\le\pi$.
Thus $\chi_{\mathrm I}>1$; combined with $\alpha_{\mathrm I}>1$, this shows
that both costs are dominated by the characteristic-function queries.
Substitution into \cref{prop:oracle-to-tv} therefore gives total running time
\begin{equation}
  \widetilde O_{\Delta,\beta,\lambda_-,\lambda_+}\left(
    \frac{n^{3+3\alpha_{\mathrm I}+2\chi_{\mathrm I}}}
         {\varepsilon^{1+2\alpha_{\mathrm I}+2\chi_{\mathrm I}}}
  \right).
  \label{eq:ising-cd-total-time}
\end{equation}
This is polynomial in $n$ and $1/\varepsilon$ for the fixed model parameters.
The empty graph returns zero in constant time.
The zero-output characterization in \cref{prop:oracle-to-tv} shows that the
output is zero if and only if
$\lambda^\mu=\lambda^\nu$.
\end{proof}

\paragraph{AI Disclosure}
We used the GPT-5.6 Sol model to assist in designing the overall
methodology underlying the deterministic approximation algorithm and developing the corresponding mathematical analysis.  This materially influenced
Sections~3--7.  The authors critically reviewed and independently verified
all material retained in the manuscript, including the algorithms, proofs,
citations, and bibliographic entries, and take full responsibility for the
correctness and originality of the paper.
\bibliographystyle{alpha}
\bibliography{refs}

\appendix

\section{Proof of the constructive Jackson lemma}
\label{app:constructive-jackson}

\begin{proof}[Proof of \cref{lem:constructive-jackson}]
We first place $F$ on a circle without increasing its Lipschitz constant.
Let
\[
  u=\frac{F(-n)+F(n)}2
\]
and define, on $[-2n,2n]$,
\[
  \widetilde F(x)=
  \begin{cases}
    u+\dfrac{x+2n}{n}\bigl(F(-n)-u\bigr),
       &-2n\leq x\leq-n,\\[6pt]
    F(x),&-n\leq x\leq n,\\[4pt]
    F(n)+\dfrac{x-n}{n}\bigl(u-F(n)\bigr),
       &n\leq x\leq2n,
  \end{cases}
\]
and then extend $\widetilde F$ $4n$-periodically.  Both added linear pieces
have slope
\[
  \frac{F(-n)-F(n)}{2n},
\]
whose absolute value is at most one.  The pieces join continuously, and
$\widetilde F(-2n)=\widetilde F(2n)=u$; hence the periodic extension is
$1$-Lipschitz.  It also takes values in $[0,2n]$ and agrees with $F$ on
$[-n,n]$.

We now make the Jackson construction explicit.  Let
$q=\left\lfloor\frac K2\right\rfloor+1$ and $d=2q-2\leq K$.
Following Jackson's construction \cite{jackson1912approximation}, define the
normalized fourth-power Jackson kernel
\begin{equation}
  \mathcal J_q(y)
  =\frac1{c_0}
    \left(\frac{\sin(\pi qy)}{\sin(\pi y)}\right)^4,
  \qquad -\frac12\leq y\leq\frac12,
  \label{eq:jackson-kernel}
\end{equation}
where $c_0=\frac{q(2q^2+1)}3$ and removable singularities are filled in by continuity.  Indeed, the
Fej\'er identity gives
\[
  \left(\frac{\sin(\pi qy)}{\sin(\pi y)}\right)^2
  =\sum_{r\in\mathbb Z}a_r\e^{2\pi\iu ry},
  \qquad
  a_r=(q-\abs r)_+.
\]
Consequently, the fourth power in \eqref{eq:jackson-kernel} has Fourier
coefficients
\[
  c_k=\sum_{r\in\mathbb Z}a_r a_{k-r},
  \qquad \abs{k}\leq d,
\]
and no other nonzero coefficients.  Its constant coefficient is
\[
  c_0=\sum_r a_r^2
     =q^2+2\sum_{s=1}^{q-1}s^2
     =\frac{q(2q^2+1)}3,
\]
consistent with the definition of $c_0$ above.  Since
\[
  \int_{-1/2}^{1/2}\e^{2\pi\iu ky}\,\mathrm dy
  =\begin{cases}1,&k=0,\\0,&k\neq0,\end{cases}
\]
integrating a finite Fourier series over one period extracts its constant
coefficient.  Hence
\[
  \int_{-1/2}^{1/2}
    \left(\frac{\sin(\pi qy)}{\sin(\pi y)}\right)^4\,\mathrm dy
  =c_0,
  \qquad
  \int_{-1/2}^{1/2}\mathcal J_q(y)\,\mathrm dy=1.
\]
Thus $\mathcal J_q$ is an even, nonnegative, normalized trigonometric
polynomial of degree $d$.  Moreover, $c_k\geq0$, while Cauchy--Schwarz
gives
\[
  c_k=\sum_r a_r a_{k-r}
  \leq\left(\sum_r a_r^2\right)^{1/2}
       \left(\sum_r a_{k-r}^2\right)^{1/2}
  =c_0.
\]
Consequently,
\[
  \mathcal J_q(y)=\sum_{\abs{k}\leq d}\frac{c_k}{c_0}
                    \e^{2\pi\iu ky}.
\]
Extend these Fourier multipliers by zero up to degree $K$, and set
\[
  m_k=
  \begin{cases}
    c_k/c_0,&\abs{k}\leq d,\\
    0,&d<\abs{k}\leq K
  \end{cases}.
\]
Then $0\leq m_k\leq1$ for every $\abs{k}\leq K$.
The convolution formula for $c_k$ also allows all these rational multipliers
to be computed exactly in $O(K^2)$ arithmetic operations.

For $\abs{k}\leq K$, let
\[
  \mathcal F_k(\widetilde F)
  =\frac1{4n}\int_{-2n}^{2n}
      \widetilde F(z)\e^{-\iu t_kz}\,\mathrm dz,
  \qquad
  \alpha_k=m_k\mathcal F_k(\widetilde F).
\]
We next verify explicitly that these coefficients give the Jackson
convolution.  The factor $4n$ below identifies the period-one $y$-circle
with the period-$4n$ circle of $\widetilde F$.  Since the Fourier series of
$\mathcal J_q$ is finite, we may integrate it term by term.  For every
$\abs{k}\leq d$, the substitution $z=x-4ny$, together with
$t_k=\pi k/(2n)$, gives
\begin{align*}
 &\int_{-1/2}^{1/2}
   \widetilde F(x-4ny)\e^{2\pi\iu ky}\,\mathrm dy\\
 &\qquad=
   \frac{\e^{\iu t_kx}}{4n}
   \int_{x-2n}^{x+2n}
     \widetilde F(z)\e^{-\iu t_kz}\,\mathrm dz.
\end{align*}
The integrand in the last integral is $4n$-periodic, because
\[
  \widetilde F(z+4n)=\widetilde F(z),
  \qquad
  \e^{-\iu t_k(z+4n)}
  =\e^{-\iu t_kz}\e^{-2\pi\iu k}
  =\e^{-\iu t_kz}.
\]
Its integral over $[x-2n,x+2n]$ therefore equals its integral over
$[-2n,2n]$.  It follows that
\[
  \int_{-1/2}^{1/2}
   \widetilde F(x-4ny)\e^{2\pi\iu ky}\,\mathrm dy
  =\mathcal F_k(\widetilde F)\e^{\iu t_kx}.
\]
Multiplying by $m_k$ and summing yields
\begin{equation}
\begin{aligned}
  \int_{-1/2}^{1/2}
       \widetilde F(x-4ny)\mathcal J_q(y)\,\mathrm dy
  &=\sum_{\abs{k}\leq d}
       m_k\mathcal F_k(\widetilde F)\e^{\iu t_kx}\\
  &=\sum_{\abs{k}\leq K}\alpha_k\e^{\iu t_kx}
   =P(x),
\end{aligned}
  \label{eq:jackson-convolution}
\end{equation}
Since $m_{-k}=m_k$ and
$\mathcal F_{-k}(\widetilde F)=\overline{\mathcal F_k(\widetilde F)}$,
$P$ is real-valued.

We now explain why the first absolute moment of the kernel controls the
error.  Normalization of $\mathcal J_q$ allows us to subtract
$\widetilde F(x)$ inside the convolution:
\begin{align*}
  P(x)-\widetilde F(x)
  &=\int_{-1/2}^{1/2}
    \bigl(\widetilde F(x-4ny)-\widetilde F(x)\bigr)
    \mathcal J_q(y)\,\mathrm dy.
\end{align*}
Since $\widetilde F$ is $1$-Lipschitz and $\mathcal J_q\geq0$,
\begin{equation}
  \abs{P(x)-\widetilde F(x)}
  \leq4n\int_{-1/2}^{1/2}\abs y\,\mathcal J_q(y)\,\mathrm dy.
  \label{eq:error-by-first-moment}
\end{equation}
It therefore remains to estimate this first absolute moment.  For
$0<\abs y\leq1/2$,
\[
  \abs{\frac{\sin(\pi qy)}{\sin(\pi y)}}
  \leq\min\left\{q,\frac1{2\abs y}\right\},
\]
by the geometric-sum formula and the bounds
$\abs{\sin(\pi qy)}\leq1$ and
$\abs{\sin(\pi y)}\geq2\abs y$.  Since $\mathcal J_q$ is even, splitting at
$y=1/(2q)$ gives
\begin{align*}
  \int_{-1/2}^{1/2}\abs y\,\mathcal J_q(y)\,\mathrm dy
  &\leq\frac2{c_0}\left(
      \int_0^{1/(2q)}yq^4\,\mathrm dy
      +\int_{1/(2q)}^{1/2}\frac{\mathrm dy}{16y^3}
    \right).
\end{align*}
The two integrals are
\[
  \int_0^{1/(2q)}yq^4\,\mathrm dy
  =q^4\left[\frac{y^2}{2}\right]_0^{1/(2q)}
  =\frac{q^2}{8}
\]
and
\[
  \int_{1/(2q)}^{1/2}\frac{\mathrm dy}{16y^3}
  =\frac1{16}\left[-\frac1{2y^2}\right]_{1/(2q)}^{1/2}
  =\frac{q^2-1}{8}.
\]
Consequently, since $c_0=q(2q^2+1)/3\geq2q^3/3$,
\begin{align*}
  \int_{-1/2}^{1/2}\abs y\,\mathcal J_q(y)\,\mathrm dy
  &\leq\frac2{c_0}\left(\frac{q^2}{8}
      +\frac{q^2-1}{8}\right)\\
  &=\frac{2q^2-1}{4c_0}
   \leq\frac{q^2}{2c_0}
   \leq\frac3{4q}.
\end{align*}
Substituting this estimate into \eqref{eq:error-by-first-moment} and using
$q\geq K/2$ gives
\[
  \abs{P(x)-\widetilde F(x)}
  \leq\frac{3n}{q}
  \leq\frac{6n}{K}
  \leq\frac{\varepsilon}{4}.
\]
This proves the approximation bound on $[-n,n]$.  Also,
$\abs{\mathcal F_k(\widetilde F)}\leq2n$ and $m_k\leq1$, so
\[
  \sum_{\abs{k}\leq K}\abs{\alpha_k}
  \leq2n(2K+1)
  \leq6nK,
\]
which proves \eqref{eq:jackson-guarantees}.

We finally give the deterministic coefficient computation.  Cache
$F(-n)$ and $F(n)$, and let $N$ be the smallest power of two satisfying
\begin{equation}
  N\geq
  \max\left\{
    2K+1,
    \left\lceil
      \frac{8n(1+\pi K)(2K+1)}{\varepsilon}
    \right\rceil
  \right\}.
  \label{eq:number-of-samples}
\end{equation}
Set
\[
  x_j=-2n+\left(j+\frac12\right)\frac{4n}{N},
  \qquad 0\leq j<N.
\]
Define
\[
  \widehat{\mathcal F}_k
  =\frac1N\sum_{j=0}^{N-1}
      \widetilde F(x_j)\e^{-\iu t_kx_j},
  \qquad
  \widehat\alpha_k=m_k\widehat{\mathcal F}_k.
\]
The integrand
$g_k(x)=\widetilde F(x)\e^{-\iu t_kx}$ is Lipschitz with constant at most
\[
  1+2n\abs{t_k}\leq1+\pi K.
\]
For completeness, on each subinterval of length $h=4n/N$, the midpoint
error for an $L$-Lipschitz function is at most
$Lh^2/4$.  After summing the $N$ errors and dividing by $4n$, the normalized
midpoint error is at most $nL/N$.  Therefore
\[
  \abs{\widehat{\mathcal F}_k
        -\mathcal F_k(\widetilde F)}
  \leq\frac{n(1+\pi K)}{N}
  \leq\frac{\varepsilon}{8(2K+1)}.
\]
Since $m_k\leq1$, summing this bound over $\abs{k}\leq K$ proves
\eqref{eq:jackson-coefficient-error}.

After caching $F(-n)$ and $F(n)$, evaluating the $N$ sample values
$\widetilde F(x_j)$ costs $O(NT_F(n))$ time.  Directly computing
$\widehat{\mathcal F}_k$ requires $O(N)$ operations for each of the
$2K+1$ frequencies, so all midpoint sums cost $O(KN)$ operations.
Computing all Jackson multipliers $m_k$ from the convolution formula
for $c_k$ costs $O(K^2)$ operations: there are $O(K)$ coefficients,
each with $O(K)$ summands.  Forming
$\widehat\alpha_k=m_k\widehat{\mathcal F}_k$ then costs another $O(K)$
operations.  Alternatively,
\[
  \e^{-\iu t_kx_j}
  =(-1)^k\e^{-\pi\iu k/N}\e^{-2\pi\iu kj/N},
\]
so all sums are fixed phase multiples of entries of one length-$N$ DFT\@.
Since $N$ is a power of two, one length-$N$ FFT computes this DFT in
$O(N\log N)$ operations.  The condition $N\geq2K+1$ ensures that the
required frequencies $-K,\ldots,K$ are distinct modulo $N$.
Extracting the corresponding DFT entries, dividing by $N$, applying
the fixed phase factors, and multiplying by $m_k$ costs $O(K)$
additional operations.  These $O(K)$ postprocessing costs are absorbed
by the other terms.  Thus direct summation takes
\begin{equation}
  O\bigl(NT_F(n)+KN+K^2\bigr)
  =O\left(
      \frac{n^3}{\varepsilon^3}T_F(n)
      +\frac{n^4}{\varepsilon^4}
    \right)
  \label{eq:direct-time}
\end{equation}
operations, where the three terms account for sampling, evaluating all
midpoint sums, and computing the multipliers, respectively.  Using an
FFT replaces the midpoint-sum cost $O(KN)$ by the DFT cost
$O(N\log N)$, giving
\begin{equation}
  O\bigl(NT_F(n)+N\log N+K^2\bigr)
  =O\left(
      \frac{n^3}{\varepsilon^3}
      \left[T_F(n)+\log\!\left(\frac{2n}{\varepsilon}\right)\right]
    \right).
  \label{eq:fft-time}
\end{equation}
Here we used $K=\Theta(n/\varepsilon)$ and
\eqref{eq:number-of-samples}, which give
\[
  N=\Theta\left(\frac{nK^2}{\varepsilon}\right)
   =\Theta\left(\frac{n^3}{\varepsilon^3}\right).
\]
The multiplier cost $O(K^2)$ is absorbed by $O(KN)$ in the direct bound
and by $O(N\log N)$ in the FFT bound.
This completes the proof.
\end{proof}

\section{Proof of the multivariate Lee--Yang approximation lemma}
\label{app:lee-yang-complex-oracle}

\begin{proof}[Proof of \cref{lem:lee-yang-complex-oracle}]
Let $n=\abs{V(H)}$.  If $n=0$, return $1$ in constant time.  If $\beta=1$, then
\[
  Z_H(1,z)=\prod_{v\in V(H)}(1+z_v),
\]
so the partition function can be computed exactly in $O(n)$ operations.
Hence assume that
$0<\beta<1$, and define the univariate polynomial
\[
  Q(t)=Z_H\bigl(\beta,(tz_v)_{v\in V(H)}\bigr).
\]
Thus $Q(0)=1$ and $Q(1)=Z_H(\beta,z)$.  If $\abs{t}<1/r$, then
$\abs{tz_v}<1$ for every $v$.  The multivariate Lee--Yang theorem, applied
separately to each connected component and combined with spin reversal, gives
\[
  Q(t)\ne0 \qquad\text{for }\abs{t}<1/r
\]
\cite[Theorem~4.2 and Eq.~(13)]{liuSinclairSrivastava2017ising}.

The coefficient algorithm in the proof of
\cite[Theorem~3.1]{liuSinclairSrivastava2017ising} applies to $Q$ by adding
one singleton hyperedge $\{v\}$ at each vertex, with unary factor
\[
  \varphi_{\{v\}}(-)=1,\qquad \varphi_{\{v\}}(+)=z_v.
\]
Keeping the original Ising edge factors and taking the common vertex
activity to be $t$ gives precisely $Q(t)$.  All factors satisfy the
normalization $\varphi_e(-,\ldots,-)=1$ required by the coefficient algorithm.
Adding these singleton hyperedges increases the maximum degree to at most
$\Delta+1$, while the maximum hyperedge size remains at most $2$.
The labelled-insect enumeration handles the distinct fixed weights without
changing the enumeration bound.  Consequently, the first $m$ coefficients of
$Q$ can be computed using $n^{O(1)}\exp(O_\Delta(m))$ arithmetic operations.
The Newton identities compute the first $m$ inverse power sums of its zeros
from these coefficients in another $O(m^2)$ operations.  In particular,
there are constants $c_\Delta,C_\Delta\ge1$, depending only on $\Delta$, for
which the combined cost is
\[
  O_{\Delta,\beta}\bigl(n^{c_\Delta}\exp(C_\Delta m)\bigr).
\]
Here $C_\Delta$ is enlarged if necessary to absorb the Newton-identity cost.

Let $\alpha_1,\ldots,\alpha_n$ be the zeros of $Q$.  Since
$\abs{\alpha_j}\ge1/r$, write $p_k=\sum_{j=1}^n\alpha_j^{-k}$.  The Taylor
approximation
\[
  T_m=-\sum_{k=1}^m\frac{p_k}{k}
\]
satisfies
\[
  \abs{\log Q(1)-T_m}
  \le n\sum_{k>m}\frac{r^k}{k}
  \le \frac{nr^{m+1}}{1-r},
\]
where the logarithm is the branch normalized by $\log Q(0)=0$.  Taking
\[
  m=\left\lceil
    \frac{\log\bigl(2n/((1-r)\varepsilon)\bigr)}{\log(1/r)}
  \right\rceil
  =O_r\bigl(\log(n/\varepsilon)\bigr)
\]
makes the last quantity at most $\varepsilon/2$.  Therefore, for
$\widehat Z_H(\beta,z)=\exp(T_m)$,
\[
  \frac{\abs{\widehat Z_H(\beta,z)-Z_H(\beta,z)}}
       {\abs{Z_H(\beta,z)}}
  \le \exp(\varepsilon/2)-1
  \le \varepsilon.
\]
Forming $T_m$ costs $O(m)$ further operations, and its exponential costs one
elementary-function evaluation.  Substituting the chosen value of $m$ gives
the total bound
\[
  O_{\Delta,\beta,r}\left(
    n^{c_\Delta+C_\Delta/\log(1/r)}
    \varepsilon^{-C_\Delta/\log(1/r)}
  \right).
\]
Thus one may take
\[
  a_{\mathrm{LY}}=c_\Delta+\frac{C_\Delta}{\log(1/r)},
  \qquad
  b_{\mathrm{LY}}=\max\left\{1,\frac{C_\Delta}{\log(1/r)}\right\}.
\]
These constants are at least one and depend only on $\Delta$ and $r$.
Together with the cases $n=0$ and $\beta=1$, this proves the stated bound
$T_Z(n,\varepsilon)
=O_{\Delta,\beta,r}((n+1)^{a_{\mathrm{LY}}}\varepsilon^{-b_{\mathrm{LY}}})$.
\end{proof}

\end{document}